\RequirePackage{etoolbox}
\csdef{input@path}{%
 {sty/}
 {figs/}
}%

\documentclass[ba,noinfoline]{imsart_mod}
\setattribute{journal}{name}{}
\setattribute{journal}{url}{}

\usepackage{amsthm}
\usepackage{amsmath}
\usepackage{natbib}
\usepackage[colorlinks,citecolor=blue,urlcolor=blue,filecolor=blue,backref=page]{hyperref}
\usepackage{graphicx}
\usepackage{booktabs} 

\startlocaldefs
\numberwithin{equation}{section}
\theoremstyle{plain}
\newtheorem{thm}{Theorem}[section]
\endlocaldefs

\usepackage{float}
\usepackage{amsbsy}
\usepackage{amstext}
\usepackage{amssymb}
\usepackage{subfig}
\providecommand{\tabularnewline}{\\}
\floatstyle{ruled}
\newfloat{algorithm}{tbp}{loa}
\providecommand{\algorithmname}{Algorithm}
\floatname{algorithm}{\protect\algorithmname}
\newtheorem{prop}[thm]{\protect\propositionname}

\theoremstyle{plain}

\numberwithin{equation}{section}
\usepackage{algorithm,algpseudocode}

\providecommand{\propositionname}{Proposition}

\everypar{\looseness=-1000}

\begin{document}

\begin{frontmatter}
\title{Spatial 3D Mat\'{e}rn priors for fast whole-brain fMRI analysis}
\runtitle{Spatial 3D Mat\'{e}rn priors for fMRI}
\begin{aug}
\author{\fnms{Per} \snm{Sid\'{e}n}\thanksref{addr1,t1}\ead[label=e1]{per.siden@liu.se}},
\author{\fnms{Finn} \snm{Lindgren}\thanksref{addr2}\ead[label=e2]{finn.lindgren@ed.ac.uk}},
\author{\fnms{David} \snm{Bolin}\thanksref{addr3}\ead[label=e3]{ david.bolin@kaust.edu.sa}},
\author{\fnms{Anders} \snm{Eklund}\thanksref{addr1,addr4}\ead[label=e4]{anders.eklund@liu.se}}
\and
\author{\fnms{Mattias} \snm{Villani}\thanksref{addr1,addr5}\ead[label=e5]{mattias.villani@gmail.com}}

\runauthor{P. Sid\'{e}n, F. Lindgren, D. Bolin, A. Eklund and M. Villani}

\address[addr1]{Division of Statistics and Machine Learning, Dept. of Computer and Information Science, Link\"{o}ping University, SE-581 83 Link\"{o}ping, Sweden.
    \printead{e1} 
    \printead{e4}
    \printead{e5}
}
\address[addr2]{School of Mathematics, The University of Edinburgh, James Clerk Maxwell Building, The King's Building, Peter Guthrie Tait Road, Edinburgh, EH9 3FD, United Kingdom.
    \printead{e2}
}
\address[addr3]{CEMSE Division, King Abdullah University of Science and Technology, Saudi Arabia.
    \printead{e3}
}
\address[addr4]{Division of Medical Informatics, Dept. of Biomedical Engineering and Center for Medical Image Science and Visualization (CMIV), Link\"{o}ping University, SE-581 83 Link\"{o}ping, Sweden.
}
\address[addr5]{Department of Statistics, Stockholm University, SE-106 91 Stockholm, Sweden.
}
\thankstext{t1}{Corresponding author.}
\end{aug}

\begin{abstract}
Bayesian whole-brain functional magnetic resonance imaging (fMRI) analysis with three-dimensional spatial smoothing priors has been shown to produce state-of-the-art activity maps without pre-smoothing the data. The proposed inference algorithms are computationally demanding however, and the proposed spatial priors have several less appealing properties, such as being improper and having infinite spatial range. We propose a statistical inference framework for whole-brain fMRI analysis based on the class of Mat\'{e}rn covariance functions. The framework uses the Gaussian Markov random field (GMRF) representation of possibly anisotropic spatial Mat\'{e}rn fields via the stochastic partial differential equation (SPDE) approach of \citet{Lindgren2011}. This allows for more flexible and interpretable spatial priors, while maintaining the sparsity required for fast inference in the high-dimensional whole-brain setting. We develop an accelerated stochastic gradient descent (SGD) optimization algorithm for empirical Bayes (EB) inference of the spatial hyperparameters. Conditionally on the inferred hyperparameters, we make a fully Bayesian treatment of the brain activity. The Mat\'{e}rn prior is applied to both simulated and experimental task-fMRI data and clearly demonstrates that it is a more reasonable choice than the previously used priors, using comparisons of activity maps, prior simulation and cross-validation.
\end{abstract}

\begin{keyword}
\kwd{spatial priors}
\kwd{Gaussian Markov random fields}
\kwd{fMRI}
\kwd{spatiotemporal modeling}
\kwd{efficient computation}
\end{keyword}

\end{frontmatter}

\section{Introduction}\label{sec:intro}

Functional magnetic resonance imaging (fMRI) is a noninvasive technique
for making inferences about the location and magnitude of neuronal
activity in the living human brain. fMRI has provided neuroscientists
with countless new insights on how the brain operates \citep{Lindquist2008}.
By observing changes in blood oxygenation in a subject during an experiment,
a researcher can apply statistical methods such as the general linear
model (GLM) \citep{Friston1995a} to draw conclusions regarding task-related
brain activations.

fMRI data can be seen as a sequence of three-dimensional images collected
over time, where each image can be divided into a large number of
voxels. A problem with the GLM approach and many of its successors
is that the model is mass-univariate, that is, it analyses each voxel
separately and ignores the inherent spatial dependencies between neighboring
brain regions. Normally, this is accounted for by pre-smoothing data
and using post-correction of multiple hypothesis testing, but this
strategy is unsatisfactory from a modeling perspective and has been
shown to lead to spurious results in many cases \citep{Eklund2016}.

One of the earliest Bayesian spatial smoothing priors for neuroimaging
is the two-dimensional prior in slice-wise analysis proposed by \citet{pennyEtAlSpatialPrior2005}.
The spatial prior on the activity coefficients reflect the prior knowledge
that activated regions are spatially contiguous and locally homogeneous.
\citet{pennyEtAlSpatialPrior2005} use the variational Bayes (VB)
approach to approximate the posterior distribution of the activations.
\citet{Sid??n2017} extend that prior to the 3D case and propose a
fast Markov Chain Monte Carlo (MCMC) method and an improved VB approach, that is empirically shown to give negligible error compared to MCMC.

In this paper, we show how the spatial priors used in the previous
articles can be seen as special cases of the Gaussian Markov random
field (GMRF) representation of Gaussian fields of the Mat\'{e}rn class,
using the stochastic partial differential equation (SPDE) approach
presented in \citet{Lindgren2011}. The Mat\'{e}rn family of covariance
functions, attributed to \citet{Matern1960} and popularized by \citet{Handcock1993},
is seeing increasing use in spatial statistical modeling. It is also
a standard choice for Gaussian process (GP) priors in machine learning
\citep{Seeger2004}. In his practical suggestions for prediction of
spatial data, \citet{Stein1999} notes that the properties of a spatial
field depends strongly on the local behavior of the field and that
this behavior is unknown in practice and must be estimated from the
data. Moreover, some commonly used covariance functions, for example
the Gaussian (also known as the squared exponential), do not provide
enough flexibility with regard to this local behavior and Stein summarizes
his suggestions with ``\textit{Use the Mat\'{e}rn model}''. Using the
Matern prior on large-scale 3D data such as fMRI data is computationally
challenging, however, in particular with MCMC. We present a fast Bayesian
inference framework to make Stein's appeal feasible in practical work.

Even though the empirical spatial auto-correlation functions of raw
fMRI data seem more fat-tailed than a Gaussian \citep{Eklund2016,Cox2017},
standard practice has traditionally been to pre-smooth data using a
Gaussian kernel with reference to the matched filter theorem. The
Gaussian covariance function has also been used directly in the model
as a spatial GP prior \citep{Groves2009}, but using the standard
GP formulation results in a dense covariance matrix which becomes
too computationally expensive to invert even with only a few thousand
voxels. For this reason, much work on spatial modeling of fMRI data
has been using GMRFs instead, see for example \citet{Gossl2001,Woolrich2004c,pennyEtAlSpatialPrior2005,Harrison2010,Sid??n2017}.
GMRFs have the property of having sparse precision matrices, which
make them computationally very fast to use, but do not always correspond
to simple covariance functions, especially the intrinsic GMRFs often
used as priors, whose precision matrices are not invertible \citep{Isham2004}.
A different branch of Bayesian spatial models for fMRI has considered
selecting active voxels as a variable selection problem, modeling
the spatial dependence between the activity indicators rather than
between the activity coefficients \citep[see, among others][]{Smith2007,Vincent2010a,Lee2014,Zhang2014,Bezener2018}.
These articles mostly use Ising priors or GMRF priors squashed through
a cumulative distribution function (CDF) for the indicator dependence,
which also gives sparsity. However, these priors are rarely defined
over the whole brain, but are applied independently to parcels or
slices, probably due to computational costs. The SPDE approach of
\citet{Lindgren2011} has been applied to fMRI data before, slice-wise
by \citet{Yue2014a}, and on the sphere by \citet{Mejia2019} after
transforming the volumetric data to the cortical surface. In both
cases integrated nested Laplace approximations (INLA) \citep{Solis-Trapala2009}
were used for approximating the posterior, which is efficient but
presently cumbersome to apply directly to volumetric fMRI data, as
the \texttt{R-INLA} R-package currently lacks support for three-dimensional
data.

Our paper makes a number of contributions. First, we develop a fast
Bayesian inference algorithm that allows us to use spatial three-dimensional
whole-brain priors of the Mat\'{e}rn class on the activity coefficients, for which previous MCMC and VB approaches are not computationally feasible. The algorithm applies empirical Bayes (EB)
to optimize the hyperparameters of the spatial prior and the parameters
of the autoregressive noise model, using an accelerated version of
stochastic gradient descent (SGD). The link to the Mat\'{e}rn covariance
function gives the spatial hyperparameters nice interpretations, in
terms of range and marginal variance of the corresponding Gaussian
field. Given the maximum a posteriori (MAP) values of the optimized
parameters, we make a fully Bayesian treatment of the main parameters
of interest, that is, the activity coefficients, and compute brain
activity posterior probability maps (PPMs). The convergence of the
optimization algorithm is established and the resulting EB posterior
is compared to the exact MCMC posterior for the prior used in \citet{Sid??n2017},
showing the results to be extremely similar. Second, we develop an anisotropic version of the Matérn 3D prior. The anisotropic prior allows the spatial dependence to vary in the
$x$-, $y$- and $z$-direction, and we choose a parameterization
such that the new parameters do not the affect the marginal
variance of the field. Third, we apply the proposed Mat\'{e}rn priors to both simulated and real fMRI
datasets, and compare with the prior used in \citet{Sid??n2017} by
observing differences in the PPMs, by examining the plausibility
of new random samples of the different spatial priors, and using cross-validation (CV) on left-out voxels to assess the predictive performance, both in terms of point predictions and predictive uncertainty. Collectively, our demonstration
strongly suggests that the higher level of smoothness is more reasonable
for fMRI data, and also indicates that the second order Matérn prior (see the definition in Section~\ref{subsec:Spatial-prior-on-activations}) is more sensible
than its intrinsic counterpart.

We begin by reviewing the model of \citet{pennyEtAlSpatialPrior2005}, also examined in \citet{Sid??n2017},
and introducing the different spatial priors and associated hyperpriors
in Section~\ref{sec:Model}. In Section~\ref{sec:Methods}, we derive
the optimization algorithm for the EB method, and describe the
PPM computation. Experimental and simulation results are shown in
Section~\ref{sec:Results}. Section~\ref{sec:Discussion-and-future}
contains conclusions and recommendations for future work. The more
mathematical details of the model and priors, the derivation
of the gradient and approximate Hessian used in the SGD optimization
algorithm, and the CV framework are given in the supplementary material. 

The new methods in this article have been implemented and added to
the BFAST3D extension to the SPM software, available at \url{http://www.fil.ion.ucl.ac.uk/spm/ext/\#BFAST3D}.

\section{Model and priors\label{sec:Model}}

The model can be divided into three parts: (i) the measurement model,
which consists of a regression model that relates the observed blood
oxygen level dependent (BOLD) signal in each voxel to the experimental
paradigm and nuisance regressors, and a temporal noise model (Section~\ref{subsec:Measurement-model}),
(ii) the spatial prior that models the dependence of the regression
parameters between voxels (Sections~\ref{subsec:Spatial-prior-on-activations}
and \ref{subsec:Anisotropic-spatial-prior}), and (iii) the priors
on the spatial hyperparameters and noise model parameters (Sections~\ref{subsec:Hyperparameter-priors}
and \ref{subsec:Noise-model-priors}).

\subsection{Measurement model\label{subsec:Measurement-model}}

The single-subject fMRI-data is collected in a $T\times N$ matrix
$\mathbf{Y}$, with $T$ denoting the number of volumes collected
over time and $N$ the number of voxels. The experimental paradigm
is represented by the $T\times K$ design matrix $\mathbf{X}$, with
$K$ regressors representing for example the hemodynamic response
function (HRF) convolved with the binary time series of task events.
The model can be written as $\mathbf{Y}=\mathbf{X}\mathbf{W}+\mathbf{E}$,
where $\mathbf{W}$ is a $K\times N$ matrix of regression coefficients
and \textbf{$\mathbf{E}$} is a $T\times N$ matrix of error terms.
We will also work with the equivalent vectorized formulation $\mathbf{y}=\mathbf{\bar{\mathbf{X}}\boldsymbol{\beta}}+\mathbf{e}$,
where $\mathbf{y}=\text{vec}\left(\mathbf{Y}^{T}\right)$, $\bar{\mathbf{X}}=\mathbf{X}\otimes\mathbf{I}_{N}$,
$\boldsymbol{\beta}=\text{vec}\left(\mathbf{W}^{T}\right)$ and $\mathbf{e}=\text{vec}\left(\mathbf{E}^{T}\right)$.
The error terms are modeled as Gaussian and independent across voxels,
possibly following voxel-specific $P$th order AR models, described
by the $N\times1$ vector $\boldsymbol{\lambda}$ of noise precisions
and the $P\times N$ matrix $\mathbf{A}$ of AR parameters. For the
ease of presentation we will in what follows only consider the special
case $P=0$, that is, error terms that are independent across both
time and voxels, and treat the more general case in the supplementary material.

We can divide our parameters into three groups: $\boldsymbol{\beta}$,
$\boldsymbol{\theta}_{n}$ and $\boldsymbol{\theta}_{s}$. Here,  $\boldsymbol{\beta}$
describes the brain activity coefficients which we are mainly interested
in, $\boldsymbol{\theta}_{n}=\left\{ \boldsymbol{\lambda},\mathbf{A}\right\} $
are parameters of the noise model, and $\boldsymbol{\theta}_{s}$
are spatial hyperparameters that will be introduced in the next subsection.

\subsection{Spatial prior on activations\label{subsec:Spatial-prior-on-activations}}

We assume spatial, three-dimensional GMRF priors \citep{Isham2004,Sid??n2017}
for the regression coefficients, which are independent across regressors,
that is, we assume $\boldsymbol{\beta}|\boldsymbol{\theta}_{s}\sim\mathcal{N}\left(\mathbf{0},\mathbf{Q}^{-1}\right)$.
Here $\mathbf{Q}=\underset{k\in\left\{ 1,\ldots,K\right\} }{\text{blkdiag}}\left[\mathbf{Q}_{k}\right]$ is a $KN\times KN$ block diagonal matrix with the $N\times N$
matrix $\mathbf{Q}_{k}$ as the $k$th block. The vector $\boldsymbol{\theta}_{s}=\left\{ \boldsymbol{\theta}_{s,1},\ldots,\boldsymbol{\theta}_{s,K}\right\} $
contains the spatial hyperparameters that the different $\mathbf{Q}_{k}$ depend on. The precision matrices $\mathbf{Q}_{k}$ may be chosen
differently for different $k$. In this paper, we construct the different $\mathbf{Q}_{k}$ using the SPDE approach \citep{Lindgren2011}, which allows for sparse GMRF representations of Mat\'{e}rn fields. An overview of the different priors can be seen in Table~\ref{tab:Precision-matrices}, and are described in more detail below.

\begin{table}
\caption{Summary of the spatial priors used and their precision matrices. The global shrinkage (GS) prior is spatially independent, while the intrinsic conditional autoregression (ICAR), Matérn (M) and anisotropic Matérn (A-M) can be seen as GMRF representations of generalized Mat\'{e}rn fields.\label{tab:Precision-matrices}}
\begin{tabular}{llll}
\toprule
Spatial prior & $\alpha$ & $\kappa$ & Precision matrix\tabularnewline
\midrule
GS & - & - & $\tau^{2}\mathbf{I}$\tabularnewline[\doublerulesep]
ICAR$\left(1\right)$ & $1$ & $=0$ & $\tau^{2}\mathbf{G}$\tabularnewline[\doublerulesep]
M$\left(1\right)$ & $1$ & $>0$ & $\tau^{2}\mathbf{K},\,\,\,\mathbf{K}=\kappa^{2}\mathbf{I}+\mathbf{G}$\tabularnewline[\doublerulesep]
ICAR$\left(2\right)$ & $2$ & $=0$ & $\tau^{2}\mathbf{G}^{T}\mathbf{G}$\tabularnewline[\doublerulesep]
M$\left(2\right)$ & $2$ & $>0$ & $\tau^{2}\mathbf{K}^{T}\mathbf{K},\,\,\,\mathbf{K}=\kappa^{2}\mathbf{I}+\mathbf{G}$\tabularnewline[\doublerulesep]
\noalign{\vskip\doublerulesep}
A-M$(2)$ & $2$ & $>0$ & $\tau^{2}\mathbf{K}^{T}\mathbf{K},\,\,\,\mathbf{K}=\kappa^{2}\mathbf{I}+h_{x}\mathbf{G}_{x}+h_{y}\mathbf{G}_{y}+h_{z}\mathbf{G}_{z}$\tabularnewline[\doublerulesep]
\bottomrule 
\end{tabular}%
\end{table}

\citet{Sid??n2017} focus on the unweighted graph Laplacian prior
$\mathbf{Q}_{k}=\tau^{2}\mathbf{G}$ which we refer to here as the
ICAR$\left(1\right)$ (first-order intrinsic conditional autoregression) prior.
The matrix $\mathbf{G}$ is defined by
\begin{equation}
\ G_{i,j}=\begin{cases}
n_{i} & ,\,\text{for\,\,}i=j\\
-1 & ,\,\text{for\,\,}i\sim j\\
0 & ,\,\text{otherwise,}
\end{cases}\label{eq:DiscreteLaplaceG}
\end{equation}
where $i\sim j$ means that $i$ and $j$ are adjacent voxels and
$n_{i}$ is the number of voxels adjacent to voxel $i$. The ICAR$\left(1\right)$
prior can be derived from the local assumption that $x_{i}-x_{j}\sim\mathcal{N}(0,\tau^{-2})$,
for all unordered pairs of adjacent voxels $(i,j)$, where $\mathbf{x}$
denotes the GMRF \citep{Isham2004}. Thus, one can see that $\tau^{2}$
controls how much the field can vary between neighboring voxels, where
large values of $\tau^{2}$ enforces a field that is spatially smooth.
The ICAR$(1)$ prior is default in the SPM software for Bayesian fMRI analysis.
The second-order ICAR$\left(2\right)$ prior is a more smooth alternative, corresponding
to a similar local assumption for the second-order differences, and
has been used earlier for fMRI analysis in 2D \citep{pennyEtAlSpatialPrior2005}.
The ICAR priors can be extended by adding $\kappa^{2}$ to the
diagonal of $\mathbf{G}$ as in the right hand column of Table~\ref{tab:Precision-matrices},
and when $\kappa>0$ we refer to these as M$\left(\alpha\right)$
($\alpha$-order Mat\'{e}rn) priors. The reason for this is the SPDE link established
by \citet{Lindgren2011}. For example, the M$\left(2\right)$ prior
can be seen as the solution $\mathbf{u}$ to
\begin{equation}
\tau\left(\kappa^{2}\mathbf{I}+\mathbf{G}\right)\mathbf{u}\sim\mathcal{N}\left(\mathbf{0},\mathbf{I}\right),\label{eq:SPDEMatrixEq}
\end{equation}
which can in turn be seen as a numerical finite difference approximation
to the SPDE
\begin{equation}
\left(\kappa^{2}-\Delta\right)^{\alpha/2}\tau u\left(\mathbf{s}\right)=\mathcal{W}\left(\mathbf{s}\right),\label{eq:SPDE}
\end{equation}
when $\alpha=2$. Here $\mathbf{s}$ denotes a point in space, $\alpha$
is a smoothness parameter, $\Delta$ is the Laplace operator, and
$\mathcal{W}\left(\mathbf{s}\right)$ is spatial white noise. Define also the smoothness parameter $\nu=\alpha-d/2$, where $d$ is the
dimension of the domain. For $\nu>0$ and $\kappa>0$, it can be shown
that a Gaussian field $u(\mathbf{s})$ is a solution to the SPDE in
Eq.~(\ref{eq:SPDE}), when it has the Mat\'{e}rn covariance function
\citep{whittle1954,whittle1963}
\begin{equation}
C(\delta)=\frac{\sigma^{2}}{2^{\nu-1}\Gamma(\nu)}\left(\kappa\delta\right)^{\nu}K_{\nu}\left(\kappa\delta\right),\label{eq:MaternCovFunc}
\end{equation}
where $\delta$ is the Euclidean distance between two points in $\mathbb{R}^{d}$,
$K_{\nu}$ is the modified Bessel function of the second kind and
\begin{equation}
\sigma^{2}=\frac{\Gamma\left(\nu\right)}{\Gamma\left(\nu+d/2\right)\left(4\pi\right)^{d/2}\tau^{2}\kappa^{2\nu}}\label{eq:MaternMarginalVariance}
\end{equation}
is the marginal variance of the field $u(\mathbf{s})$. As $d=3$
in our case, for $\alpha=2$ we have $\nu=1/2$ which is a special
case where the Mat\'{e}rn covariance function is the same as the exponential
covariance function. In this paper we also consider the SPDE when
$\kappa=0$ or $\nu=-1/2$, in which case the solutions no longer
have Mat\'{e}rn covariance, but are still well-defined random measures,
and we will refer to them as generalized Mat\'{e}rn fields.

We also define $\mathbf{K}=\kappa^{2}\mathbf{I}+\mathbf{G}$, in which
case the solution to Eq.~(\ref{eq:SPDEMatrixEq}) is $\mathbf{u}\sim\mathcal{N}\left(\mathbf{0},\left(\tau^{2}\mathbf{K}\mathbf{K}\right)^{-1}\right)$,
which is largely the same as the solution obtained in \citet{Lindgren2011}
using the finite-element method when the triangle basis points are
placed at the voxel locations, apart for some minor differences at
the boundary. We use the same definition as \citet{Lindgren2011}
for the range $\rho=\sqrt{8\nu}/\kappa$, for $\kappa>0$ and $\nu>0$,
which is a distance for which two points in the field have correlation
near to $0.13$. This reveals an important interpretation of the ICAR$(2)$ prior, since this can be seen as a special case of M$(2)$ with $\kappa=0\Leftrightarrow \rho=\infty$, that is, infinite range. For other values of $\alpha$, similar simple discrete solutions of the SPDE are also available. In particular, for $\alpha=1$
we have $\mathbf{u}\sim\mathcal{N}\left(\mathbf{0},\left(\tau^{2}\mathbf{K}\right)^{-1}\right)$.
Extensions to higher integer values of $\alpha$ such as $\alpha=3,\,4,\ldots$
are straightforward in theory \citep{Lindgren2011}, but will result
in less sparse precision matrices $\mathbf{Q}_{k}$ and thereby longer
computing times, and more involved gradient expressions for the parameter
optimization in Section~\ref{subsec:Parameter-optimization}.

For each choice of $\mathbf{Q}_{k}$, we have spatial hyperparameters
$\boldsymbol{\theta}_{s,k}=\left\{ \tau_{k}^{2},\kappa_{k}^{2}\right\} $,
which will normally be estimated from data. For regressors not related
to the brain activity, that is, head motion regressors and voxel intercepts,
we do not use a spatial prior, but instead a global shrinkage (GS)
prior with precision matrix $\mathbf{Q}_{k}=\tau_{k}^{2}\mathbf{I}$. We could here infer
$\tau_{k}^{2}$ from the data, but will normally fix it to some small
value, for example $\tau_{k}^{2}=10^{-12}$, which gives a non-informative
prior that provides some numerical stability.

\subsection{Anisotropic spatial prior\label{subsec:Anisotropic-spatial-prior}}

The SPDE approach makes it possible to fairly easily construct anisotropic
priors, for example using a SPDE of the form

\begin{equation}
\left(\kappa^{2}-h_{x}\frac{\partial^{2}}{\partial x^{2}}-h_{y}\frac{\partial^{2}}{\partial y^{2}}-h_{z}\frac{\partial^{2}}{\partial z^{2}}\right)^{\alpha/2}\tau u\left(\mathbf{s}\right)=\mathcal{W}\left(\mathbf{s}\right),\label{eq:AnisoSPDE}
\end{equation}
with $h_{z}$ defined as $h_{z}=\frac{1}{h_{x}h_{y}}$ for identifiability.
For $\alpha=2$, this SPDE has a finite-difference solution with precision
matrix $\tau^{2}\mathbf{K}\mathbf{K}$, with $\mathbf{K}$ now defined
as $\mathbf{K}=h_{x}\mathbf{G}_{x}+h_{y}\mathbf{G}_{y}+h_{z}\mathbf{G}_{z}+\kappa^{2}\mathbf{I}$.
Here, $\mathbf{G}_{x}$ is defined as in Eq.~(\ref{eq:DiscreteLaplaceG}),
after redefining the neighbors as being only the adjacent
voxels in the $x$ direction. $\mathbf{G}_{y}$ and $\mathbf{G}_{z}$
are defined correspondingly, so that $\mathbf{G}=\mathbf{G}_{x}+\mathbf{G}_{y}+\mathbf{G}_{z}$.
When using this prior for regressor $k$ we have four parameters,
$\boldsymbol{\theta}_{s,k}=\left\{ \tau_{k}^{2},\kappa_{k}^{2},h_{x,k},h_{y,k}\right\} $.
The new parameters $h_{x}$ and $h_{y}$ allows for different relative
length scales of the spatial dependence in the $x$-, $y$- and $z$-direction,
which is reasonable considering the data might not have voxels of
equal size in all dimensions and the data collection is normally not
symmetric with respect to the three axes. Conveniently, $h_{x}=h_{y}=1$
gives the standard isotropic Mat\'{e}rn field defined earlier.
\begin{prop}
For $\alpha>d/2$, the anisotropic field $u$ defined in Eq.~(\ref{eq:AnisoSPDE}) on $\mathbb{R}^d$
has the marginal variance defined in Eq.~(\ref{eq:MaternMarginalVariance}),
 and the variance thus does not depend on $h_{x}$
and $h_{y}$. Furthermore, $Cov\left(u(\mathbf{s}),u(\mathbf{t})\right)=C\left(\sqrt{\left(\mathbf{s}-\mathbf{t}\right)^{T}\mathbf{H}^{-1}\left(\mathbf{s}-\mathbf{t}\right)}\right)$,
where $\mathbf{H}$ is a diagonal matrix with diagonal $\left(h_{x},h_{y},1/\left(h_{x}h_{y}\right)\right)^{T}$,
and C$\left(\delta\right)$ is the isotropic Mat\'{e}rn covariance function
defined in Eq.~(\ref{eq:MaternCovFunc}) with $\nu=\alpha-d/2$.\label{prop:AnisotropicMargVariance}
\end{prop}

\begin{proof}
We show the covariance formula first, and then the statement about
the marginal variance follows as $Cov\left(u(\mathbf{s}),u(\mathbf{s})\right)=C\left(\sqrt{\mathbf{0}^{T}\mathbf{H}^{-1}\mathbf{0}}\right)=C\left(0\right)$.
By using a certain definition of the Fourier transform, the spectral
density of $u$ in the anisotropic SPDE in Eq.~(\ref{eq:AnisoSPDE})
is 
\begin{equation}
S\left(\boldsymbol{\omega}\right)=\frac{1}{\left(2\pi\right)^{d}}\frac{1}{\tau^{2}\left(\kappa^{2}+\boldsymbol{\omega}^{T}\mathbf{H}\boldsymbol{\omega}\right)^{\alpha}},\label{eq:spectralDensity}
\end{equation}
so the covariance function can be written as
\begin{equation}
Cov\left(u(\mathbf{s}),u(\mathbf{t})\right)=\int_{\mathbb{R}^{d}}\frac{1}{\left(2\pi\right)^{d}}\frac{1}{\tau^{2}\left(\kappa^{2}+\boldsymbol{\omega}^{T}\mathbf{H}\boldsymbol{\omega}\right)^{\alpha}}e^{-i\boldsymbol{\omega}^{T}\left(\mathbf{s}-\mathbf{t}\right)}d\boldsymbol{\omega}.\label{eq:covIntegral}
\end{equation}
An isotropic field $v$ can be written as an anisotropic field with
$\mathbf{H}=\mathbf{I}$, so its covariance function for $\delta=\left\Vert \mathbf{s}-\mathbf{t}\right\Vert _{2}$
is
\begin{equation}
Cov\left(v(\mathbf{s}),v(\mathbf{t})\right)=\int_{\mathbb{R}^{d}}\frac{1}{\left(2\pi\right)^{d}}\frac{1}{\tau^{2}\left(\kappa^{2}+\boldsymbol{\omega}^{T}\boldsymbol{\omega}\right)^{\alpha}}e^{-i\boldsymbol{\omega}^{T}\left(\mathbf{s}-\mathbf{t}\right)}d\boldsymbol{\omega}.\label{eq:covFuncCovIntegral}
\end{equation}
On the other hand,
\begin{equation}
Cov\left(v(\mathbf{H}^{-1/2}\mathbf{s}),v(\mathbf{H}^{-1/2}\mathbf{t})\right)=\int_{\mathbb{R}^{d}}\frac{1}{\left(2\pi\right)^{d}}\frac{1}{\tau^{2}\left(\kappa^{2}+\boldsymbol{\omega}^{T}\boldsymbol{\omega}\right)^{\alpha}}e^{-i\boldsymbol{\omega}^{T}\left(\mathbf{H}^{-1/2}\mathbf{s}-\mathbf{H}^{-1/2}\mathbf{t}\right)}d\boldsymbol{\omega}\label{eq:cov_s-t}
\end{equation}
\[
=\int_{\mathbb{R}^{d}}\frac{1}{\left(2\pi\right)^{d}}\frac{1}{\tau^{2}\left(\kappa^{2}+\mathbf{z}^{T}\mathbf{H}\mathbf{z}\right)^{\alpha}}e^{-i\mathbf{z}^{T}\left(\mathbf{s}-\mathbf{t}\right)}\det\left(\mathbf{H}^{1/2}\right)d\mathbf{z}
\]
where the last step used the variable substitution $\boldsymbol{\omega}=\mathbf{H}^{1/2}\mathbf{z}$.
Since $\det\left(\mathbf{H}^{1/2}\right)=\sqrt{h_{x}\cdot h_{y}\cdot1/\left(h_{x}h_{y}\right)}=1$,
the last expression equals that in Eq.~(\ref{eq:covIntegral}). So
\begin{equation}
Cov\left(u(\mathbf{s}),u(\mathbf{t})\right)=Cov\left(v(\mathbf{H}^{-1/2}\mathbf{s}),v(\mathbf{H}^{-1/2}\mathbf{t})\right)=C\left(\sqrt{\left(\mathbf{s}-\mathbf{t}\right)^{T}\mathbf{H}^{-1}\left(\mathbf{s}-\mathbf{t}\right)}\right),\label{eq:endofProof}
\end{equation}
using that $\sqrt{\left(\mathbf{s}-\mathbf{t}\right)^{T}\mathbf{H}^{-1}\left(\mathbf{s}-\mathbf{t}\right)}=\left\Vert \mathbf{H}^{-1/2}\mathbf{s}-\mathbf{H}^{-1/2}\mathbf{t}\right\Vert _{2}$.
\end{proof}
Proposition~\ref{prop:AnisotropicMargVariance} implies that changing
$h_{x}$ or $h_{y}$ does not affect the marginal variance of the
field. This is convenient because it means that the anisotropic parameterization
does not change the interpretation of $\tau^{2}$ and $\kappa^{2}$,
apart from that $\rho=\sqrt{8\nu}/\kappa$ will now be the (in some sense) average range in the $x$-, $y$- and $z$-direction. Thus, we can use the
same priors for $\tau^{2}$ and $\kappa^{2}$ as in the isotropic case. By putting log-normal
priors on $h_{x}$ and $h_{y}$, as explained in the next subsection,
we get priors that are symmetric with respect to the $x$-, $y$-
and $z$-direction.

\subsection{Hyperparameter priors\label{subsec:Hyperparameter-priors}}

We will now specify priors for the spatial hyperparameters $\boldsymbol{\theta}_{s}=\left\{ \boldsymbol{\theta}_{s,1},\ldots,\boldsymbol{\theta}_{s,K}\right\} $,
which we let be independent across the different regressors $k$. For brevity,
we drop subindexing with respect to $k$ in what follows.

Penalised complexity (PC) priors \citep{Simpson2017} provide a framework
for specifying weakly informative priors that penalize deviation from
a simpler base model. \citet{Fuglstad2018} showed the usefulness
of PC priors for the hyperparameters of Mat\'{e}rn Gaussian random fields,
where the base model is chosen for $\kappa^{2}$ as the intrinsic
field $\kappa^{2}=0$ and the base model for $\tau^{2}|\kappa^{2}$
is chosen as the model with zero variance, that is $\tau^{2}=\infty$
(note that our definition of $\tau^{2}$ corresponds to $\tau^{-1}$
in \citet{Fuglstad2018}). This means exponential priors for $\kappa^{d/2}$
and for $\tau^{-1}|\kappa^{2}$. The PC prior for M$(2)$ allows the
user to be weakly informative about range and standard deviation of
the spatial activation coefficient maps, by a priori controlling the
lower tail probability for the range $\mathrm{Pr}\left(\rho<\rho_{0}\right)=\xi_{1}$
and the upper tail probability for the marginal variance $\mathrm{Pr}\left(\sigma^{2}>\sigma_{0}^{2}\right)=\xi_{2}$
of the field. By default, we will set $\xi_{1}=\xi_{2}=0.05$,
$\rho_{0}$ to 2 voxel lengths and $\sigma_{0}^{2}$ corresponding
to $5\%$ probability that the marginal standard deviation of the
activity coefficients is larger than $2\%$ of the global mean signal.
See the supplementary material for full details
about the PC prior for the M$(2)$ hyperparameters.

For M$(1)$, PC priors are not straightforward to specify, since the
range and marginal variance are not available for $\nu=-1/2$ in the
continuous space, so we will instead use log-normal priors for $\tau^{2}$
and $\kappa^{2}$, as specified in the supplementary material.

For ICAR$(1)$ and ICAR$(2)$ we use the PC prior for $\tau^{2}$ for Gaussian random effects in \citet[][Section 3.3]{Simpson2017}, and we follow their suggestion for handling the  singular precision matrix. Since these spatial priors do not have a finite marginal variance, we let the PC prior control the marginal variances of $\boldsymbol{\beta}|\mathbf{V}^T \boldsymbol{\beta}=\mathbf{0}$ instead, where $\mathbf{V}$ is the nullspace of the prior precision matrix. These nullspaces are known by construction, and the variances measure deviances beyond the addition of a constant to all voxels for ICAR(1), and beyond the addition of constants and linear trends for ICAR(2). The variances are inversely proportional to $\tau^2$, and we numerically computed them through simulation using a typical brain (from the word object experiment described below) to be $\bar{\sigma}^2=0.29/\tau^2$ for ICAR(1) and $\bar{\sigma}^2=0.76/\tau^2$ for ICAR(2), on average across all voxels. We specify $\sigma_{0}^{2}$ and $\xi_{2}$ so that $\mathrm{Pr}(\bar{\sigma}^2>\sigma_{0}^{2})=\xi_{2}$ and use $\xi_{2}=0.05$ and $\sigma_{0}$ corresponding to 2\% of the global mean signal.

For the anisotropic priors we use log-normal priors for $h_{x}$ and
$h_{y}$ as 
\begin{equation}
\left[\begin{array}{c}
\log h_{x}\\
\log h_{y}
\end{array}\right]\sim\mathcal{N}\left(\mathbf{0},\sigma_{h}^{2}\left[\begin{array}{cc}
1 & -\frac{1}{2}\\
-\frac{1}{2} & 1
\end{array}\right]\right),\label{eq:anisotropicHyperPrior}
\end{equation}
which means that also $\log\left(1/\left(h_{x}h_{y}\right)\right)\sim\mathcal{N}\left(0,\sigma_{h}^{2}\right)$
with correlation $-1/2$ with $\log h_{x}$ and $\log h_{y}$. The
motivation for this prior is that it is centered at the isotropic
model $h_{x}=h_{y}=1$, and it is symmetric with respect to the $x$-,
$y$- and $z$-direction. We will use $\sigma_{h}^{2}=0.01$ as default,
which roughly corresponds to a $(0.8,1.2)$ $95\%$-interval for $h_{x}$.

\subsection{Noise model priors\label{subsec:Noise-model-priors}}

We use priors for the noise model parameters $\boldsymbol{\theta}_{n}=\left\{ \boldsymbol{\lambda},\mathbf{A}\right\}$
that are independent across voxels and across AR parameters within
the same voxel, with $\lambda_{n}\sim\Gamma\left(u_{1},u_{2}\right)$
and $A_{p,n}\sim\mathcal{N}(0,1/\tau_{A}^{2})$,
which is the same prior as in \citet{pennyEtAlSpatialPrior2005}.
Normally we use $u_{1}=10$ and $u_{2}=0.1$, which are the default
values in the SPM software and $\tau_{A}^{2}=10^{-3}$ which is the
value used in \citet{pennyEtAlSpatialPrior2005}. We have seen that
the spatial prior for the AR parameters previously used \citep{Penny2007,Sid??n2017}
gives similar results in practice, which is why we use the computationally
more simple independent prior for $\mathbf{A}$.

\section{Bayesian inference algorithm\label{sec:Methods}}

The fast MCMC algorithm in \citet{Sid??n2017} is not trivially extended
to a 3D model with a Mat\'{e}rn prior as the updating step for $\kappa_{k}^{2}$
conditional on the other parameters requires the computation of log
determinants such as $\log\left|\kappa_{k}^{2}\mathbf{I}+\mathbf{G}\right|$
for various $\kappa_{k}^{2}$. This in general requires the Cholesky
decomposition of $\kappa_{k}^{2}\mathbf{I}+\mathbf{G}$ which has
overwhelming memory and time requirements for large $N$ and would
normally not be feasible for whole-brain analysis. In addition, $\kappa_{k}^{2}$
would require some proposal density for a Metropolis-within-Gibbs-step,
as a conjugate prior is not available. The same problems apply to
the MCMC steps for $h_{x,k}$ and $h_{y,k}$ when using the anisotropic
model. The lack of conjugate priors also makes the spatial VB (SVB)
method in \citet{Sid??n2017} more complicated, as the mean-field
VB approximate marginal posterior of $\kappa_{k}^{2}$ will no longer have a simple closed form.

We instead take an EB approach and optimize the spatial and noise
model parameters $\boldsymbol{\theta}=\left\{ \boldsymbol{\theta}_{s},\boldsymbol{\theta}_{n}\right\} $,
for which we are not directly interested in the uncertainty, with
respect to the log marginal posterior $p\left(\boldsymbol{\theta}|\mathbf{y}\right)$.
Conditional on the posterior mode estimates of $\boldsymbol{\theta}$,
we then sample from the joint posterior of the parameters of interest,
the activation coefficients in $\text{\ensuremath{\beta}}$, from
which we construct posterior probability maps (PPM) of activations.
Optimizing $\boldsymbol{\theta}$ is computationally attractive as
we can use fast stochastic gradient methods (see Section~ \ref{subsec:Parameter-optimization})
tailored specifically for our problem. We also note that VB tends
to underestimate the posterior variance of the hyperparameters \citep{Bishop2006,Solis-Trapala2009,Sid??n2017}. The approximate
posterior for $\boldsymbol{\beta}$ in \citet{Sid??n2017} only depends
on the posterior mean of the hyperparameters, still it gives very small error compared to MCMC. Thus, if EB is seen
as approximating the distribution of each hyperparameter in $\boldsymbol{\theta}$
as a point mass, it might not be much of a restriction compared to
VB.

The marginal posterior of $\boldsymbol{\theta}$ can be computed by
\begin{equation}
p\left(\boldsymbol{\theta}|\mathbf{y}\right)=\left.\frac{p\left(\mathbf{y}|\boldsymbol{\beta},\boldsymbol{\theta}\right)p\left(\boldsymbol{\beta}|\boldsymbol{\theta}\right)p\left(\boldsymbol{\theta}\right)}{p\left(\boldsymbol{\beta}|\mathbf{y},\boldsymbol{\theta}\right)p\left(\mathbf{y}\right)}\right|_{\boldsymbol{\beta}=\boldsymbol{\beta}^{*}},\label{eq:marginal_posterior}
\end{equation}
for arbitrary value of $\boldsymbol{\beta}^{*}$, where all involved
distributions are known in closed form, apart from $p\left(\mathbf{y}\right)$,
but this disappears when taking the derivative of $\log p\left(\boldsymbol{\theta}|\mathbf{y}\right)$
with respect to $\theta_{i}$. In Section~\ref{subsec:Parameter-optimization},
we comprehensibly present the optimization algorithm, but leave the
finer details to the supplementary material. Given the optimal
value $\hat{\boldsymbol{\theta}}$, we will study the full joint
posterior $\boldsymbol{\beta}|\mathbf{y},\hat{\boldsymbol{\theta}}$
of activity coefficients, which is normally the main interest for
task-fMRI analysis. This distribution is a GMRF with mean $\tilde{\boldsymbol{\mu}}$
and precision matrix $\tilde{\mathbf{Q}}$, see details in the supplementary material,
and can be used for example to compute PPMs, as described in Section~\ref{subsec:PPM-computation}.

\subsection{Parameter optimization\label{subsec:Parameter-optimization}}

By using the EB approach with SGD optimization, we avoid the costly
log determinant computations needed for MCMC, since the computation
of the posterior of $\boldsymbol{\theta}$ is no longer needed. Our
algorithm instead uses the gradient of $\log p\left(\boldsymbol{\theta}|\mathbf{y}\right)$
to optimize $\boldsymbol{\theta}$, for which there is a cheap unbiased
estimate. We also use an approximation of the Hessian and other techniques
to obtain an accelerated SGD algorithm as described below.

The optimization of $\boldsymbol{\theta}$ will be carried out iteratively.
At iteration $j$ each $\theta_{i}$ is updated with some step $\Delta\theta_{i}$
as $\theta_{i}^{\left(j\right)}=\theta_{i}^{\left(j-1\right)}+\Delta\theta_{i}^{\left(j\right)}$.
Let $G\left(\theta_{i}^{\left(j-1\right)}\right)=\left.\frac{\partial}{\partial\theta_{i}}\log p\left(\boldsymbol{\theta}|\mathbf{y}\right)\right|_{\boldsymbol{\theta}=\boldsymbol{\theta}^{\left(j-1\right)}}$
denote the gradient and $H\left(\theta_{i}^{\left(j-1\right)}\right)=\left.\frac{\partial^{2}}{\partial\theta_{i}^{2}}\log p\left(\boldsymbol{\theta}|\mathbf{y}\right)\right|_{\boldsymbol{\theta}=\boldsymbol{\theta}^{\left(j-1\right)}}$
denote the Hessian for $\theta_{i}$ (note that we here use the term
Hessian to describe a single number for each $i$, rather than the
full Hessian matrix for $\boldsymbol{\theta}$ which would be too
large to consider). Ideally, one would use the Newton method with $\Delta\theta_{i}^{\left(j\right)}=-G\left(\theta_{i}^{\left(j-1\right)}\right)\slash H\left(\theta_{i}^{\left(j-1\right)}\right)$,
or at least some gradient descent method with $\Delta\theta_{i}^{\left(j\right)}=-\eta G\left(\theta_{i}^{\left(j-1\right)}\right)$,
with some learning rate $\eta$. It turns out that for our model,
this is not computationally feasible in general, since the gradient
depends on various traces on the form $\text{tr}(\tilde{\mathbf{Q}}^{-1}\mathbf{T})$
for some matrix $\mathbf{T}$ with similar sparsity structure as $\tilde{\mathbf{Q}}$.
For small problems, such traces can be computed exactly by first computing
the selected inverse $\tilde{\mathbf{Q}}^{inv}$ of $\tilde{\mathbf{Q}}$
using the Takahashi equations \citep{Takahashi1973,Rue2007,Sid??n2017},
but this is prohibitive for problems of size larger than, say, $KN>10^{5}$.
However, the Hutchinson estimator \citep{Hutchinson1990} gives a
stochastic unbiased estimate of the trace as $\text{tr}(\tilde{\mathbf{Q}}^{-1}\mathbf{T})\approx\frac{1}{N_{s}}\sum_{j=1}^{N_{s}}\mathbf{v}_{j}^{T}\tilde{\mathbf{Q}}^{-1}\mathbf{T}\mathbf{v}_{j}$,
where each $\mathbf{v}_{j}$ is a $N\times1$ vector with independent
random elements $1$ or $-1$ with equal probability. This can be computed without computing $\tilde{\mathbf{Q}}^{-1}$, hence, we can
obtain an unbiased estimate of the gradient, which enables SGD. Using a learning rate $\eta^{\left(j\right)}$ with
the decay properties $\sum_{j}\left(\eta^{\left(j\right)}\right)^{2}<\infty$
and $\sum_{j}\eta^{\left(j\right)}=\infty$ guarantees convergence
to a local optimum \citep{RobbinsMonro,asmussen2007stochastic}.

\begin{algorithm}
\begin{algorithmic}[1]
\Require Initial values $\boldsymbol{\theta}_0$ and parameters $N_{iter},\gamma_{1},\gamma_{2},\eta_{mom},\left\{ \eta^{\left(j\right)}\right\} _{j=1}^{N_{iter}},N_{Polyak},N_{s}$
\For{$j=1$ to $N_{iter}$}
\vspace{1mm}        
\State Estimate the gradient $G\left(\theta_{i}^{\left(j-1\right)}\right)=\left.\frac{\partial}{\partial\theta_{i}}\log p\left(\boldsymbol{\theta}|\mathbf{y}\right)\right|_{\boldsymbol{\theta}=\boldsymbol{\theta}^{\left(j-1\right)}}$ for all $i$
\vspace{1mm}        
\State Estimate the approximate Hessian $\tilde{H}\left(\theta_{i}^{\left(j-1\right)}\right)=\left.E_{\boldsymbol{\beta}|\mathbf{Y},\boldsymbol{\theta}}\left[\frac{\partial^{2}\log p\left(\boldsymbol{\theta}|\mathbf{y},\boldsymbol{\beta}\right)}{\partial\theta_{i}^{2}}\right]\right|_{\boldsymbol{\theta}=\boldsymbol{\theta}^{\left(j-1\right)}}$ \newline \-\ ~~~~for all $i$
\vspace{1mm}
\State Average $\bar{G}\left(\theta_{i}^{\left(j-1\right)}\right)=\gamma_{1}\bar{G}\left(\theta_{i}^{\left(j-2\right)}\right)+\left(1-\gamma_{1}\right)G\left(\theta_{i}^{\left(j-1\right)}\right)$ for all $i$
\vspace{1mm}
\State Average $\bar{H}\left(\theta_{i}^{\left(j-1\right)}\right)=\gamma_{2}\bar{H}\left(\theta_{i}^{\left(j-2\right)}\right)+\left(1-\gamma_{2}\right)\tilde{H}\left(\theta_{i}^{\left(j-1\right)}\right)$ for all $i$
\vspace{1mm}
\State Compute $\boldsymbol{\theta}_{s}$ step sizes $\Delta\theta_{s,i}^{\left(j\right)}=\eta_{mom}\Delta\theta_{s,i}^{\left(j-1\right)}-\frac{\eta^{\left(j\right)}}{\bar{H}\left(\theta_{s,i}^{\left(j-1\right)}\right)}\bar{G}\left(\theta_{s,i}^{\left(j-1\right)}\right)$ for all $i$
\vspace{1mm}
\State Compute $\boldsymbol{\theta}_{n}$ step sizes $\Delta\theta_{n,i}^{\left(j\right)}=\eta_{n}\eta^{\left(j\right)}\bar{G}\left(\theta_{n,i}^{\left(j-1\right)}\right)$ for all $i$
\vspace{1mm}
\State Take step $\theta_{i}^{\left(j\right)}=\theta_{i}^{\left(j-1\right)}+\Delta\theta_{i}^{\left(j\right)}$ for all $i$
\vspace{1mm}
\EndFor
\vspace{1mm}
\State Return $\hat{\boldsymbol{\theta}}=\frac{1}{N_{Polyak}}\sum_{i=N_{iter}-N_{Polyak}+1}^{N_{iter}}\boldsymbol{\theta}^{\left(j\right)}$
\end{algorithmic}\caption{Parameter optimization algorithm\label{Alg:SGD}}
\end{algorithm}

To speed up the convergence of the spatial hyperparameters $\boldsymbol{\theta}_{s}$, in addition to SGD, we use an approximation of the Hessian
$\tilde{H}\left(\theta_{i}^{\left(j-1\right)}\right)=\left.E_{\boldsymbol{\beta}|\mathbf{Y},\boldsymbol{\theta}}\left[\frac{\partial^{2}\log p\left(\boldsymbol{\theta},\boldsymbol{\beta}|\mathbf{y}\right)}{\partial\theta_{i}^{2}}\right]\right|_{\boldsymbol{\theta}=\boldsymbol{\theta}^{\left(j-1\right)}}$, which improves the step length \citep{Lange1995,Bolin2018}. This
is also stochastically estimated using Hutchinson estimators of various
traces, for example\newline $\text{tr}\left(\mathbf{K}_{k}^{-1}\mathbf{K}_{k}^{-1}\right)\approx\frac{1}{N_{s}}\sum_{j=1}^{N_{s}}\mathbf{v}_{j}^{T}\mathbf{K}_{k}^{-1}\mathbf{K}_{k}^{-1}\mathbf{v}_{j}$,
where $\mathbf{K}_{k}^{-1}\mathbf{v}_{j}$ needs only to be computed
once for each $j$. The final optimization algorithm presented in
Algorithm~\ref{Alg:SGD} also uses: i) averaging over iterations
(line 4-5), which gives robustness to the stochasticity in the estimates,
ii) momentum (line 6), which gives acceleration in the relevant direction
and dampens oscillations, and iii) Polyak averaging (line 10), which
reduces the error in the final estimate of $\boldsymbol{\theta}$ by assuming that the last few iterations are just stochastic deviations from the mode. In practice, all parameters are reparametrized to be defined over the whole real line, see the supplementary material for details.

Some practical details about the optimization algorithm follow. Normally, the maximum number of iterations used is $N_{iter}=200$ the averaging parameters are $\gamma_{1}=0.2$ and $\gamma_{2}=0.9$, the momentum parameter is $\eta_{mom}=0.5$, the learning rate decreases as $\eta^{\left(j\right)}=\frac{0.9}{0.1\max\left(0,j-100\right)+1}$, the learning rate for $\boldsymbol{\theta}_{n}$ is $\eta_{n}=0.001$, we use $N_{Polyak}=10$ values for the Polyak averaging, and $N_{s}=50$ samples for the Hutchinson estimator. These parameter values led to desirable behavior when monitoring the optimization algorithm on different datasets. We initialize the noise parameters by pre-estimating the
model without the spatial prior, and the spatial parameters are normally
initialized near to the prior mean. We also start the algorithm by
running a few (normally 5) iterations of SGD with small learning rate. In each iteration, we also check the sign of the approximate Hessian to prevent steps in the direction opposite to the gradient, which could happen due to the stochasticity or local non-convexity, and change the sign if necessary. 

The computational bottleneck of the algorithm is the computation of
large matrix solves, such as $\tilde{\mathbf{Q}}^{-1}\mathbf{v}_{j}$,
involving the multiplication of the inverse of large sparse precision
matrices with a vector. This is carried out using the fast preconditioned
conjugate gradient (PCG) iterative solvers of the corresponding equation
system $\tilde{\mathbf{Q}}\mathbf{u}=\mathbf{v}_{j}$, as described in \citet{Sid??n2017},
where it is also illustrated that PCG is numerous times faster than
directly solving the equation system using the Cholesky decomposition
in these models. In addition, since the Hutchinson estimator requires many matrix solves in each iteration, these can performed in parallel on separate cores, giving great speedup.

\subsection{PPM computation\label{subsec:PPM-computation}}

PPMs are used to summarize the posterior information about active
voxels. The marginal PPM is computed for each voxel $n$ and contrast
vector $\mathbf{c}$ as $P(\mathbf{c}^{T}\mathbf{W}_{\cdot,n}>\gamma|\mathbf{y},\hat{\boldsymbol{\theta}})$,
for some activity threshold $\gamma$, recalling that $\text{vec}(\mathbf{W}^{T})=\boldsymbol{\beta}$
are the activity coefficients. Since $\boldsymbol{\beta}|\mathbf{y},\hat{\boldsymbol{\theta}}\sim\mathcal{N}(\tilde{\boldsymbol{\mu}},\tilde{\mathbf{Q}}^{-1})$
is a GMRF (see the supplementary material),
it is clear that $\mathbf{c}^{T}\mathbf{W}_{\cdot,n}|\mathbf{y},\hat{\boldsymbol{\theta}}$
is univariate Gaussian and the PPM would be simple to compute for
any \textbf{$\mathbf{c}$} if we only had access to the mean and covariance
matrix of $\mathbf{W}_{\cdot,n}|\mathbf{y},\hat{\boldsymbol{\theta}}$
for every voxel $n$. The mean is known, but the covariance matrix
is non-trivial to compute, since the posterior is parameterized using
the precision matrix. We therefore use the simple Rao-Blackwellized
Monte Carlo (simple RBMC) estimate in \citet{Siden2018} to approximate
this covariance matrix using
\begin{align}
\begin{split}
\text{Var}\left(\mathbf{W}_{\cdot,n}|\mathbf{y},\hat{\boldsymbol{\theta}}\right)=\,&\text{E}_{\mathbf{W}_{\cdot,-n}}\left[\text{Var}\left(\mathbf{W}_{\cdot,n}|\mathbf{W}_{\cdot,-n},\mathbf{y},\hat{\boldsymbol{\theta}}\right)\right]+\\&\text{Var}_{\mathbf{W}_{\cdot,-n}}\left[\text{E}\left(\mathbf{W}_{\cdot,n}|\mathbf{W}_{\cdot,-n},\mathbf{y},\hat{\boldsymbol{\theta}}\right)\right],\label{eq:simpleRBMC}
\end{split}
\end{align}
where $-n$ denotes all voxels but $n$. The first term of the right
hand side is cheaply computed as the inverse of a $K\times K$ subblock
of $\tilde{\mathbf{Q}}$. The second term is approximated by producing
$N_{RBMC}$ samples $\mathbf{W}^{\left(j\right)}$ from $\mathbf{W}|\mathbf{y},\hat{\boldsymbol{\theta}}$,
computing $\text{E}\left(\mathbf{W}_{\cdot,n}|\mathbf{W}_{\cdot,-n}^{\left(j\right)},\mathbf{y},\hat{\boldsymbol{\theta}}\right)$
analytically for each $j$, and computing the Monte Carlo approximation
of the variance. We leave out the details for brevity, but this computation
is straightforward due to the Gaussianity and computationally cheap
due to the sparsity structure of $\tilde{\mathbf{Q}}$. The PPM computation time will normally be dominated by the GMRF sampling, which is done
using the technique invented in \citet{Papandreou2010} and summarized
in \citet[Algorithm 2]{Sid??n2017}, and requires solving $N_{RBMC}$
equation systems involving $\tilde{\mathbf{Q}}$ using PCG.

\section{Results\label{sec:Results}}

This section is divided into three subsections. We start by analysing simulated fMRI data, to demonstrate the EB method's capability to estimate the true parameters, and to visualise the differences between the spatial priors in a controlled setting. We then consider real fMRI data from two different experiments, and compare the results when using different spatial priors by: inspecting the posterior activity maps, examining the plausibility of new random samples from the spatial priors, and evaluating the predictive performance using cross-validation. In the last subsection, we evaluate approximation error of the EB method by comparing to full MCMC. All computations are performed using our own Matlab code which is linked to in the end of Section~\ref{sec:intro}.

\subsection{Simulated data}

We consider a simulated dataset that is randomly generated using the anisotropic Matérn (A-M$(2)$) prior with fixed hyperparameters. The size and shape of the brain is taken from the word object dataset, described below. We first simulate four different 3D fields of activity coefficients $\boldsymbol{\beta}=\text{vec}(\mathbf{W}^T$) using four A-M$(2)$ priors with different hyperparameters. We select the hyperparameters to highlight different spatial characteristics and name the four composed conditions: \emph{Weak} (Small activation magnitude, low $\sigma$), \emph{Short range} (Short spatial range $\rho$), \emph{Long range} (Long spatial range $\rho$) and \emph{Anisotropic} ($h_x\neq1$ and $h_y\neq1$). A summary of the selected hyperparameters can be seen in Table~\ref{tab:SimulatedHyperparameters}, and one slice of the activity coefficient maps are shown in Fig.~\ref{fig:trueSimulatedMaps}.

We then use the simulated $\boldsymbol{\beta}$ coefficients to generate a time series of fMRI volumes. In order to do this, we also borrow the following variables from the word object dataset: the columns of the design matrix $\mathbf{X}$ corresponding to the HRF and intercept, the estimated values for the elements in $\boldsymbol{\beta}$ corresponding to the intercept, and estimated values for the noise variables $\boldsymbol{\lambda}$ and $\mathbf{A}$. The generated dataset has $T=100$ time points.

We use the EB method to estimate the model with the different spatial priors described in Table~\ref{tab:Precision-matrices}. The estimated hyperparameters for the A-M$(2)$ can be seen in Table~\ref{tab:SimulatedHyperparameters}. The estimates indicate that the method manages to recover the true parameter values fairly well, especially when the signal is strong (high $\sigma$) and the range is short (small $\rho$). However, when the signal is weak the estimates are more affected by the noise, and when the range is long there is bias from boundary effects since a long range is harder to infer on a limited domain. The anisotropic parameters $h_x$ and $h_y$ are correctly estimated in general, but the anisotropy is somewhat underestimated for the \emph{Anisotropic} condition, due to shrinkage from the prior. 

\begin{table}
\caption{Spatial hyperparameters of the anisotropic Matérn model (A-M$(2)$), used for the four conditions when simulating the data, and estimated values for the same data, computed using the EB method. Spatial range $\rho=2/\kappa$ (in mm),  marginal standard deviation $\sigma$ (see
Eq.~\ref{eq:MaternMarginalVariance}) and anisotropic parameters $h_x$ and $h_y$. \label{tab:SimulatedHyperparameters}}

\begin{tabular}{c|cccc|cccc}
\toprule
 & \multicolumn{4}{c}{True values} & \multicolumn{4}{c}{A-M$(2)$ estimates}\tabularnewline
\midrule
Condition & $\rho$ & $\sigma$ & $h_x$ & $h_y$ & $\rho$ & $\sigma$ & $h_x$ & $h_y$ \tabularnewline
\midrule
Weak & 18 & 1 & 1 & 1 & 15.0 & 0.96 & 1.06 & 1.01 \tabularnewline

Short range & 9 & 2 & 1 & 1 & 9.1 & 1.97 & 0.96 & 1.02 \tabularnewline

Long range & 60 & 2 & 1 & 1 & 48.7 & 1.88 & 0.90 & 1.07 \tabularnewline

Anisotropic & 18 & 1 & 0.5 & 2 & 16.9 & 1.05 & 0.60 & 1.67 \tabularnewline
\bottomrule
\end{tabular}
\end{table}

\begin{figure}
\includegraphics[width=0.75\linewidth,trim={1mm 7mm 2mm 2mm},clip]{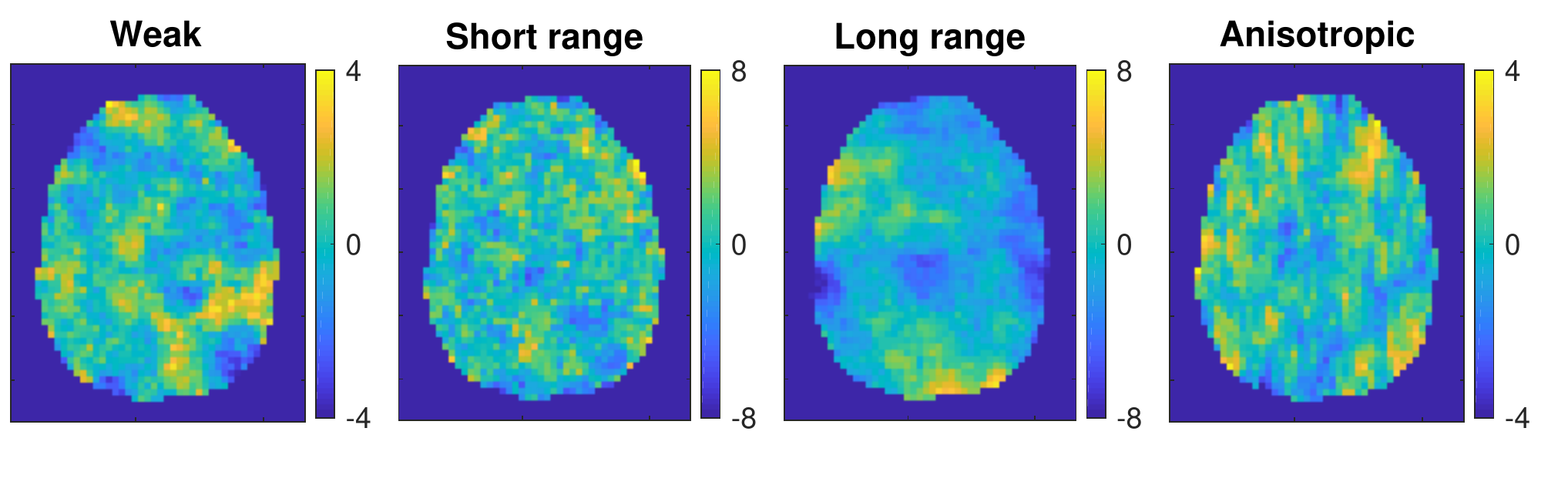}

\caption{True activity coefficients $\boldsymbol{\beta}$ for the four conditions of the simulated dataset.}
\label{fig:trueSimulatedMaps}
\end{figure}

\begin{figure}
\includegraphics[width=1\linewidth,trim={1mm 8mm 2mm 4mm},clip]{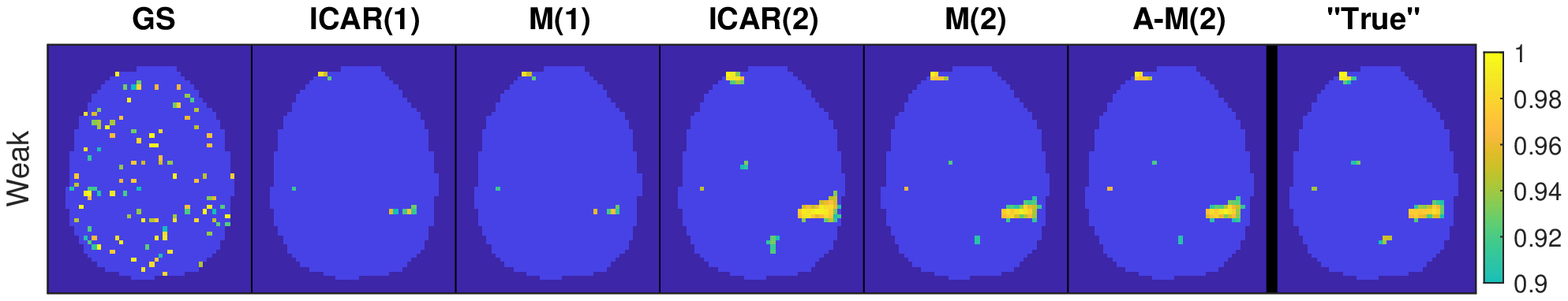}
\includegraphics[width=1\linewidth,trim={1mm 8mm 2mm 9mm},clip]{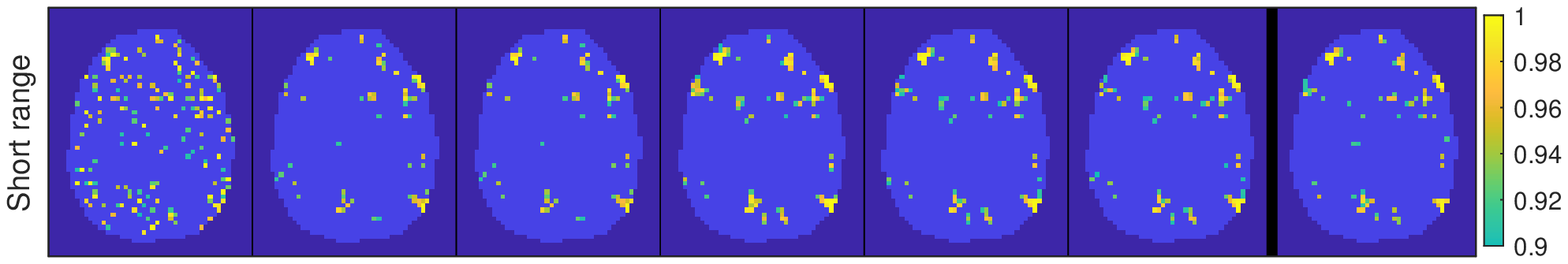}
\includegraphics[width=1\linewidth,trim={1mm 8mm 2mm 9mm},clip]{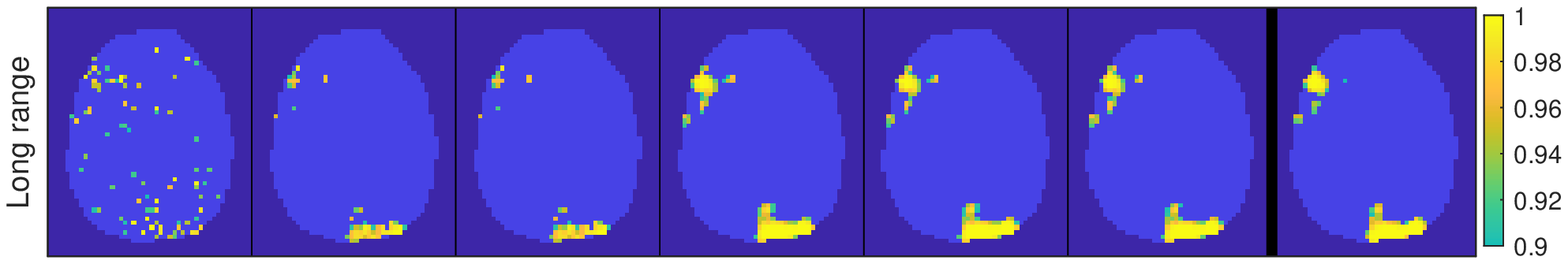}
\includegraphics[width=1\linewidth,trim={1mm 8mm 2mm 9mm},clip]{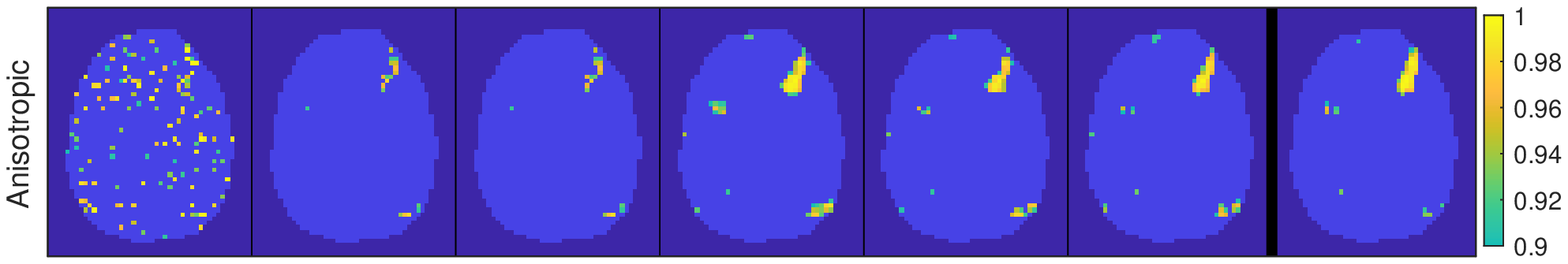}

\caption{PPMs for the four conditions of the simulated dataset, estimated with different spatial priors. The last column ``True'' shows the results when the true A-M$(2)$ hyperparameters used to generate the data is used for estimation. The PPMs show probabilities of exceeding $0.2\%$ of the global mean signal, thresholded at $0.9$. See the definition of the spatial priors in Table~\ref{tab:Precision-matrices}. The corresponding posterior means are shown in the supplementary material.}
\label{fig:simulatedMaps}
\end{figure}

Fig.~\ref{fig:simulatedMaps} shows the resulting PPMs for the different spatial priors, with hyperparameters estimated by EB, for the same slice as in Fig.~\ref{fig:trueSimulatedMaps}. The last column also shows the ``true'' PPMs obtained by using the A-M$(2)$ with the hyperparameters used to generate the data. We see how the non-spatial GS prior leads to cluttered PPMs which bear little resemblance with the true activity coefficients. We note that the first-order ICAR$(1)$ and M$(1)$ priors, with smoothness $\alpha=1$, tend to show smaller activity patterns than the second-order priors with $\alpha=2$, except for perhaps the \emph{Short range} condition. The differences between the second-order priors ICAR$(2)$, M$(2)$ and A-M$(2)$ are quite subtle, but for the \emph{Weak} and \emph{Anisotropic} conditions ICAR$(2)$ shows some signs of over-smoothing, resulting in slightly larger activity regions compared to the truth. As expected, M$(2)$ and A-M$(2)$ show little discrepancy for the first three isotropic conditions, but for the \emph{Anisotropic} condition A-M$(2)$ is to some degree closer to the truth.

\subsection{Real data}

\subsubsection{Description of the data\label{subsec:data-description}}

We evaluate the method on two different real fMRI datasets, the face
repetition dataset \citep{Henson2002} previously examined in \citet{pennyEtAlSpatialPrior2005,Sid??n2017},
and the word object dataset \citep{Duncan2009}.
The face repetition dataset is available at SPM's homepage (\url{http://www.fil.ion.ucl.ac.uk/spm/data/face\_rep/})
and the word object dataset is available at OpenNEURO (\url{https://openneuro.org/datasets/ds000107/versions/00001})
\citep{poldrack2017openfmri}. 
Both experiments have four conditions or subject tasks. Thus, the design matrix $\mathbf{X}$ for both datasets has $K=15$ columns, with column $(1,3,5,7)$ corresponding to the standard canonical HRF convolved with the different task paradigms, column $(2,4,6,8)$ corresponding to the HRF derivative, column $9$ to $14$ corresponding to head motion parameters and the last column corresponding to the intercept. 

The face repetition dataset was aqcuired during an event-related experiment, where greyscale images of non-famous and famous faces were presented to the subject for 500 ms. The four conditions in the dataset corresponds to the first and second time a non-famous or famous face was shown. The contrast studied below ``mean effect of faces'' is the average of the HRF regressors, that is $\mathbf{c}
^T \mathbf{W}_{\cdot,n}=(W_{1,n}+W_{3,n}+W_{5,n}+W_{7,n})/4$, and the presented PPMs can therefore be interpreted as showing brain regions involved in face processing. The dataset was preprocessed using the same steps as in \citet{pennyEtAlSpatialPrior2005} using
SPM12 (including motion correction, slice timing correction and normalization
to a brain template, but no smoothing), and small isolated clusters
with less than 400 voxels were removed from the brain mask. The resulting
mask has $N=57184$ voxels and there are $T=351$ volumes.

The word object experiment also has conditions that correspond to visual stimuli: written words, pictures of common objects, scrambled pictures of the same objects, and consonant letter strings, which were presented to the subject for 350 ms according to a block-related design. For the word object data, preprocessing consisted only of motion correction
and removal of isolated clusters of voxels, as the slice time information
was not available. We selected subject 10, which had relatively little
head motion, and the resulting brain mask has $N=41486$ voxels and
the number of volumes is $T=166$.

For both datasets, the voxels are
of size $3\times3\times3$ mm and the global mean signal is computed as the average value across all voxels in the brain mask and all volumes, and the activity threshold $\gamma$ used in the PPM computation is related to this quantity. 

\subsubsection{Posterior results}

We estimate the models with the different spatial priors for the two real datasets using the EB method, and present the resulting PPMs in Fig.~\ref{fig:realPPMs}.
As for the simulated dataset, we observe cluttered PPMs for the non-spatial GS prior, and in general the priors with $\alpha=1$ (ICAR$(1)$ and M$(1)$) lead to substantially smaller activity regions compared to the priors with $\alpha=2$. Given the same $\alpha$, the differences between the Matérn and ICAR priors do not seem as striking, but for the word object data, the ICAR$(2)$ prior produces an activity region in the left-hand side of the brain that is much smaller for the M$(2)$ and A-M$(2)$ priors.

\begin{figure}
\includegraphics[width=1\linewidth,trim={1mm 3mm 2mm 2mm},clip]{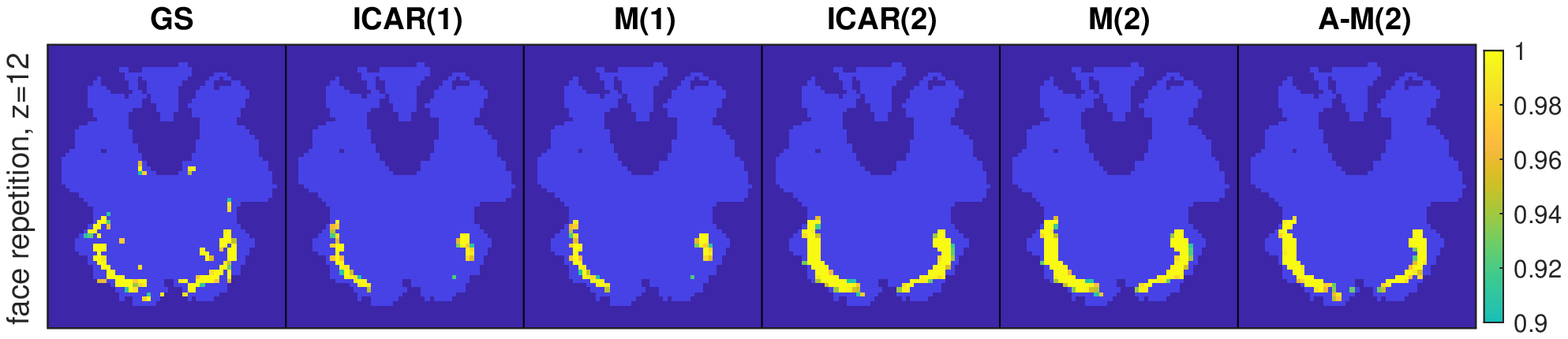}
\includegraphics[width=1\linewidth,trim={1mm 6mm 2mm 7mm},clip]{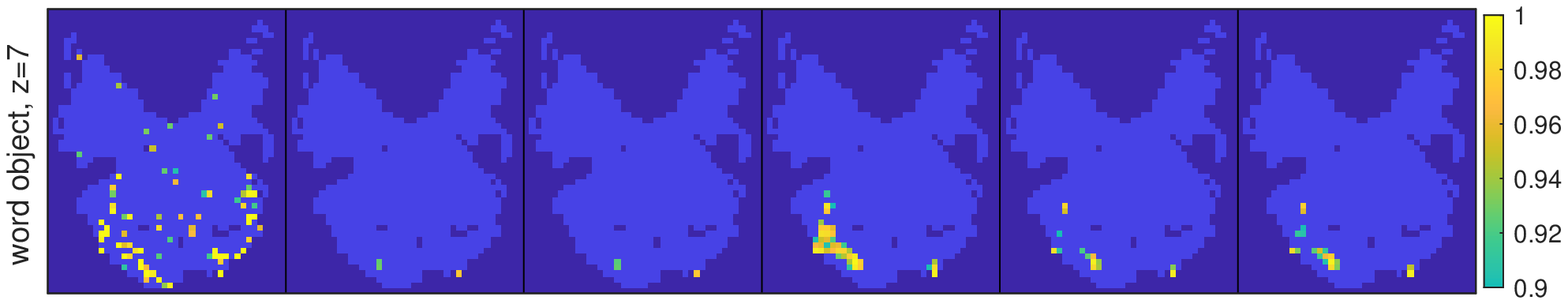}
\includegraphics[width=1\linewidth,trim={1mm 6mm 2mm 10mm},clip]{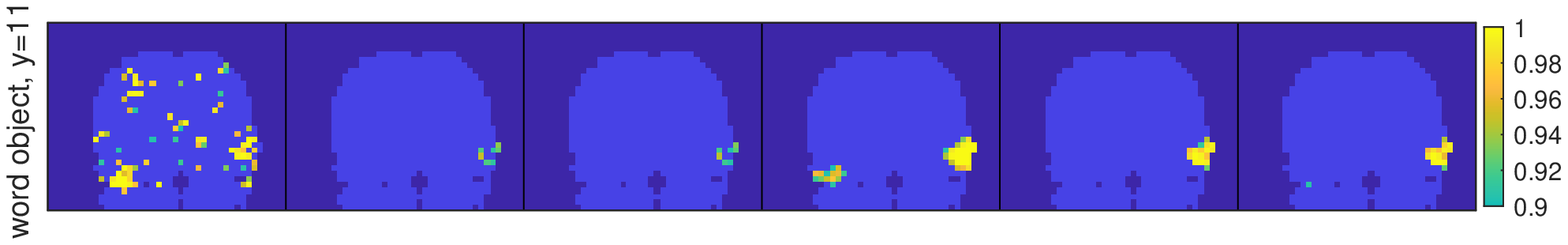}

\caption{PPMs for the two real datasets, when using different spatial priors, thresholded at $0.9$. The spatial priors are summarised in Table~\ref{tab:Precision-matrices}. The top row shows axial slice 12 of the face repetition dataset, and the middle and bottom rows show axial slice 7 and coronal slice 11 of the word object dataset. The face repetition PPMs consider the contrast ``mean effect of faces'' and show probabilities of effect sizes exceeding $1\%$
of the global mean signal. The word object PPMs consider the first condition ``Words'' and show probabilities of effect sizes exceeding $0.5\%$
of the global mean signal. The corresponding posterior means are shown in the supplementary material.}
\label{fig:realPPMs}
\end{figure}

The use of second order Matérn priors enables simple interpretations of the spatial properties of the inferred activity coefficient fields. We report the estimated hyperparameters when using the A-M$(2)$ prior for the four conditions in respective dataset in Table~\ref{tab:Hyperparameters}. The results show that the face repetition data activity patterns have longer spatial ranges (higher $\rho$) and generally larger magnitudes (higher $\sigma$), compared to the word object data. The anisotropic parameters indicate stronger dependence in the $z$-direction (between slices) for the face repetition data, while the opposite is true for the word object data.

\begin{table}
\caption{Estimated spatial hyperparameters for the A-M$(2)$ prior by the EB method, for the different datasets
and conditions. Spatial range $\rho=2/\kappa$ (in mm),  marginal standard deviation $\sigma$ (see
Eq.~\ref{eq:MaternMarginalVariance}) and anisotropic parameters $h_x$ and $h_y$. \label{tab:Hyperparameters}}

\begin{tabular}{c|cccc}
\toprule
\multicolumn{5}{c}{Face repetition data}\tabularnewline
\midrule 
Condition & $\rho$ & $\sigma$ & $h_x$ & $h_y$ \tabularnewline
\midrule 
Non-famous 1 & 62.9 & 2.36  & 0.72 & 0.75 \tabularnewline

Non-famous 2 & 58.7 & 2.40 & 0.73 & 0.73 \tabularnewline

Famous 1 & 59.5 & 2.28 & 0.70 & 0.74 \tabularnewline

Famous 2 & 47.0 & 1.98 & 0.79 & 0.68 \tabularnewline
\bottomrule
\end{tabular}~%
\hspace{3mm}
\begin{tabular}{c|cccc}
\toprule
\multicolumn{5}{c}{Word object data}\tabularnewline
\midrule 
Condition & $\rho$ & $\sigma$ & $h_x$ & $h_y$\tabularnewline
\midrule 
Words & 10.5 & 1.07 & 1.21 & 1.11\tabularnewline

Objects & 16.0 & 0.93 & 1.13 & 1.08\tabularnewline

Scrambled & 21.0 & 1.24 & 1.18 & 1.16\tabularnewline

Consonant & 11.7 & 2.09 & 1.14 & 1.23\tabularnewline
\bottomrule
\end{tabular}
\end{table}

The observed differences between the datasets could be explained by differences in the studied subjects and tasks, but is likely as well an effect from differences in scanner properties and that the slice timing and normalization preprocessing steps impose some smoothness for the face repetition data. The latter are spatial properties that would preferably be modelled in the noise rather than in the activity patterns, and we view improved, computationally efficient spatial noise models for fMRI data as important future work.

Nevertheless, the ability to flexibly estimate and indicate different spatial properties is indeed a great advantage of the second order Matérn models. Furthermore, these Matérn models correspond to exponential autocorrelation
functions, whose fat tails resemble the empirical autocorrelation
functions for fMRI data found in \citet[Supplementary Fig. 17]{Eklund2016} and \citet[Fig. 3]{Cox2017}.

\subsubsection{Prior simulation}

To better understand the meaning of the different priors in practice, Fig.~\ref{fig:priorSamples} displays samples from the spatial priors using the estimated hyperparameters for the first regressor of the different datasets. The M$(1)$ and ICAR$(1)$ priors produce fields that vary quite rapidly, while the second order priors give realizations that are more smooth. For the word object dataset we note that the short estimated range for M$(2)$ gives a sample with much faster variability than the sample from ICAR$(2)$, which looks unrealistically smooth. This illustrates the problem with using the infinite range ICAR$(2)$ prior for a dataset where the inherent range is much shorter.
\begin{figure}
\includegraphics[width=\linewidth,trim={1mm 8mm 2mm 1mm},clip]{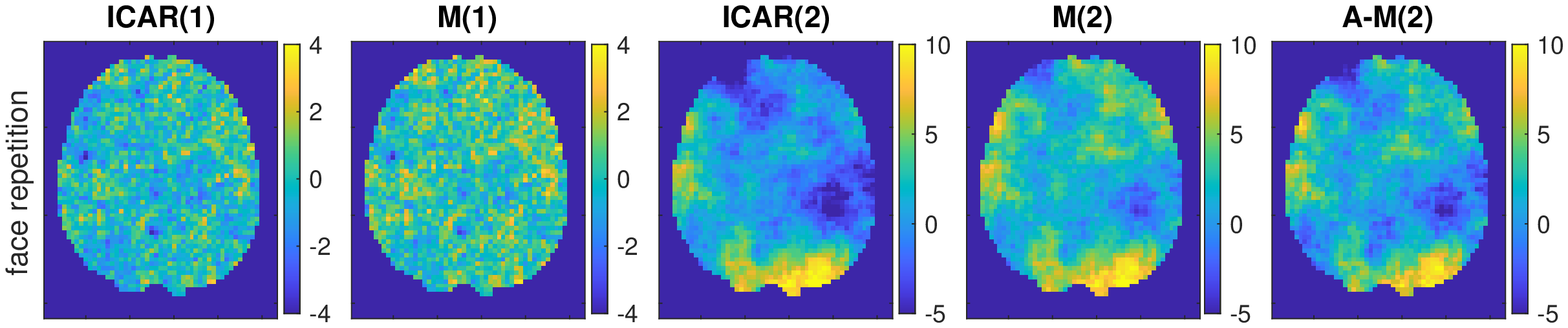}
\includegraphics[width=\linewidth,trim={1mm 8mm 2mm 8mm},clip]{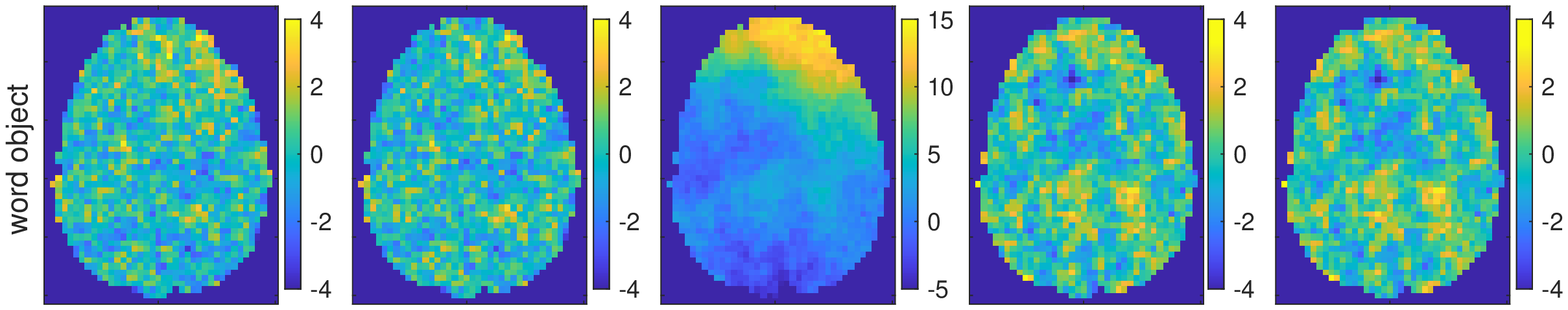}

\caption{Random samples from the different spatial priors using the estimated hyperparameters for the first regressor of the different datasets. The same seed has been used for the same $\alpha$ and dataset.}
\label{fig:priorSamples}
\end{figure}

\subsubsection{Cross-validation}

Many studies, including this one, evaluate models for fMRI data by
displaying the estimated brain activity maps and deciding whether
they look plausible or not. A more scientifically sound approach would
be to compare models based on their ability to predict the values
of unseen data points, which is the standard procedure in many other
statistical applications. The problem for fMRI data is that the main
object of interest, the set of activity coefficients $\mathbf{W}$
corresponding to activity related regressors, is not directly observable,
but only indirectly through the observed noisy BOLD signal $\mathbf{Y}$.
This makes direct comparison to ground truth activation impossible.
We will here attempt to evaluate the performance of the spatial priors
for brain activity by measuring the out-of-sample predictive performance
by computing various prediction error scores on $\mathbf{Y}$ instead.
We cannot, however, expect to find large differences between the different
priors, as only a small fraction of the signal is explained by brain
activation; most is explained by the intercept and various noise sources.

We compute the out-of-sample fit using CV while repeatedly leaving out 90\% of the voxels randomly over the whole brain, and compare the estimated and actual
signal $\mathbf{Y}$ in those voxels. In order to focus the comparison on the evaluation of the spatial priors, we must compute the errors in a slightly more cumbersome way than normal, which is explained in the supplementary material, to reduce the impact of the noise model, head motion and intercept regressors.

\begin{figure}
\includegraphics[width=1\linewidth,trim={40mm 0mm 35mm 1mm},clip]{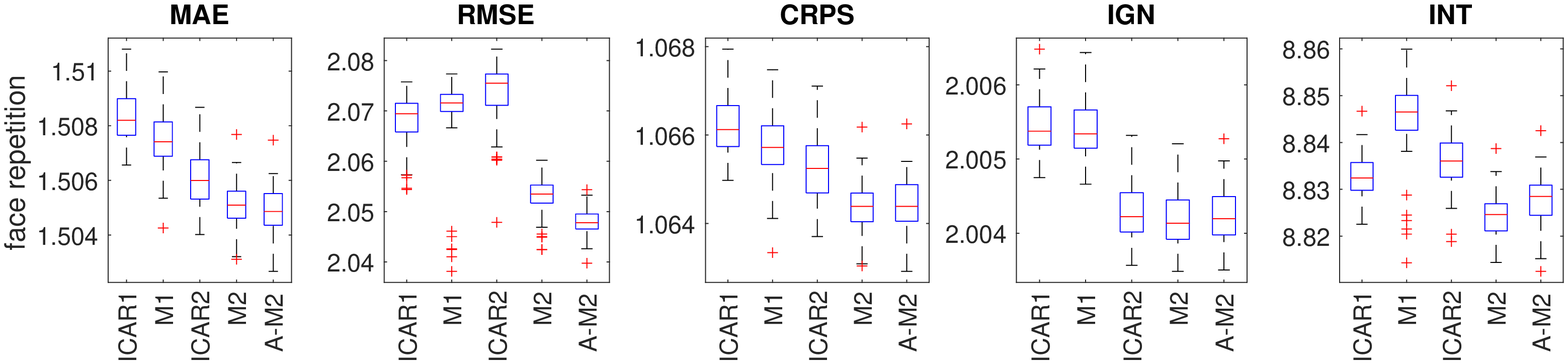}
\includegraphics[width=1\linewidth,trim={40mm 0mm 35mm 1mm},clip]{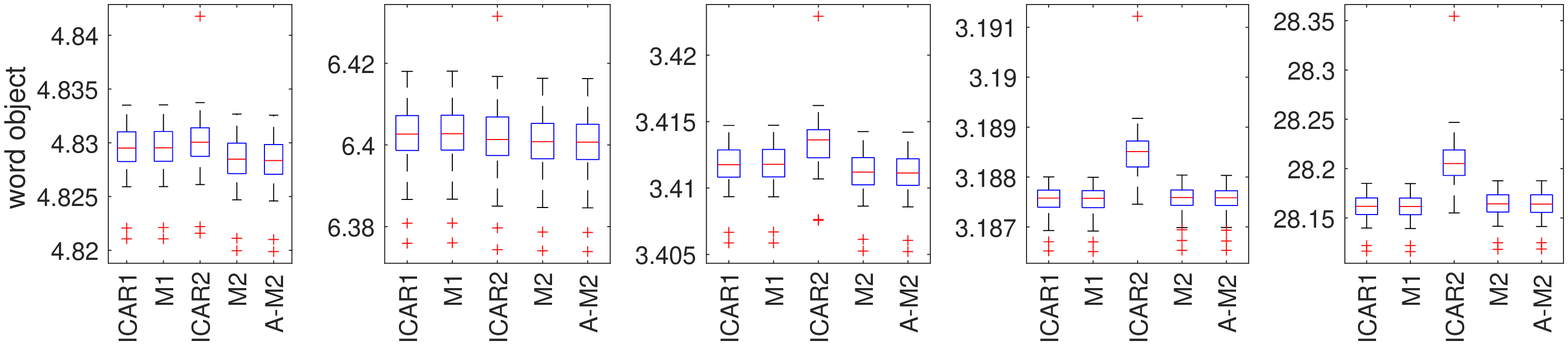}

\caption{Cross-validation scores computed on 90\% left out voxels for the two datasets, comparing the different spatial priors. The scores are computed as means across voxels, and presented in negatively oriented forms, so that smaller values are always better. The boxplots reflect the variation in 50 random sets of left out voxels. Additional results in the supplementary material shows these scores also for in-sample fit and 50\% left out voxels.}
\label{fig:BoxPlotsCV}
\end{figure}

We use the mean absolute error (MAE) and root mean square error (RMSE)
to evaluate the predicted mean of $\mathbf{Y}$ for each prior, and
the mean continuous ranked probability score (CRPS), the mean ignorance score
(IGN, also known as the logarithmic score) and the mean interval score
(INT) to evaluate the whole predictive distribution for $\mathbf{Y}$.
All these scores are all examples of proper scoring rules \citep{Gneiting2007},
which encourage the forecaster to be honest and the expected score
is maximized when the predictive distribution equals the generative
distribution of the data points. Since the predictive distribution
is Gaussian given the hyperparameters, all the scores can be computed
using simple formulas, see the supplementary material.

The results can be seen in Fig.~\ref{fig:BoxPlotsCV}. For the face repetition data, we note that the second order Matérn priors (M$(2)$ and A-M$(2)$) perform better than the other priors in all cases. For the word object data the differences between different priors are smaller, which can probably be explained by the higher noise level and shorter spatial correlation range in this dataset, but the second order Matérn priors are generally among the best. The absolute differences between the different priors may seem small, but one must remember that most of the error comes from noise that is unrelated to the brain activity, making it hard for a spatial activity prior to substantially reduce the error. The large RMSE for the ICAR$(2)$ prior for the face repetition data indicates that this prior can give relatively large out-of-sample errors, possibly due to over-smoothing.

\subsection{Evaluation of the EB method and comparison to MCMC\label{subsec:EB-evaluation}}

One of the most challenging aspects with our work has been in the development of the EB method, summarised in Algorithm~\ref{Alg:SGD}, and in finding optimization parameters that result in stable and fast convergence in the optimization of the spatial hyperparameters.
The convergence behaviour for the M$(2)$ prior is depicted in Fig.~\ref{fig:Convergence}. The hyperparameter optimization trajectories in Fig.~\ref{fig:Convergence}a suggest that the parameters reach the right level in about 100 iterations.
The results presented in this paper are all, more conservatively, after 200 iterations of optimization, but future work could include coming up with some automatic convergence criterion, based on the change of some parameters over the iterations.
Fig.~\ref{fig:Convergence}b shows how the PPM of the word object data converges. The computing time on a computing cluster, with two 8-core (16 threads) Intel Xeon E5-2660 processors at 2.2 GHz, was 3.0h until convergence (100 iterations). This time cannot be directly compared to  MCMC, as MCMC is not computationally feasible for the M$(2)$ prior, however, when using the computationally cheaper (more sparse) ICAR$(1)$ prior in the analysis below, the computation time was almost a week.

\begin{figure}
\subfloat[]{\includegraphics[width=.92\linewidth,trim={10mm 4mm 0mm 5mm},clip]{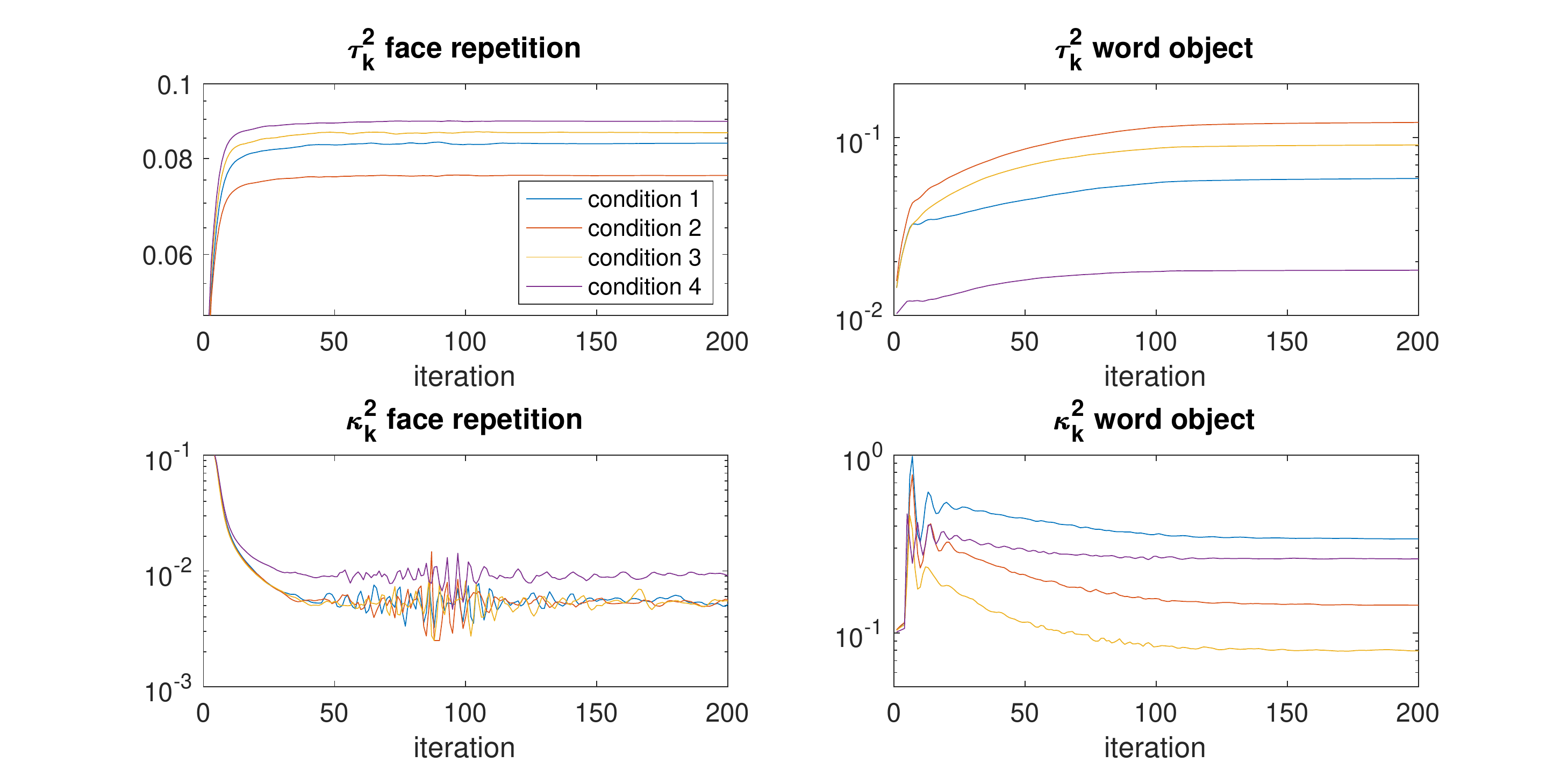}
}

\subfloat[]{\includegraphics[width=0.75\linewidth,trim={1mm 2mm 2mm 1mm},clip]{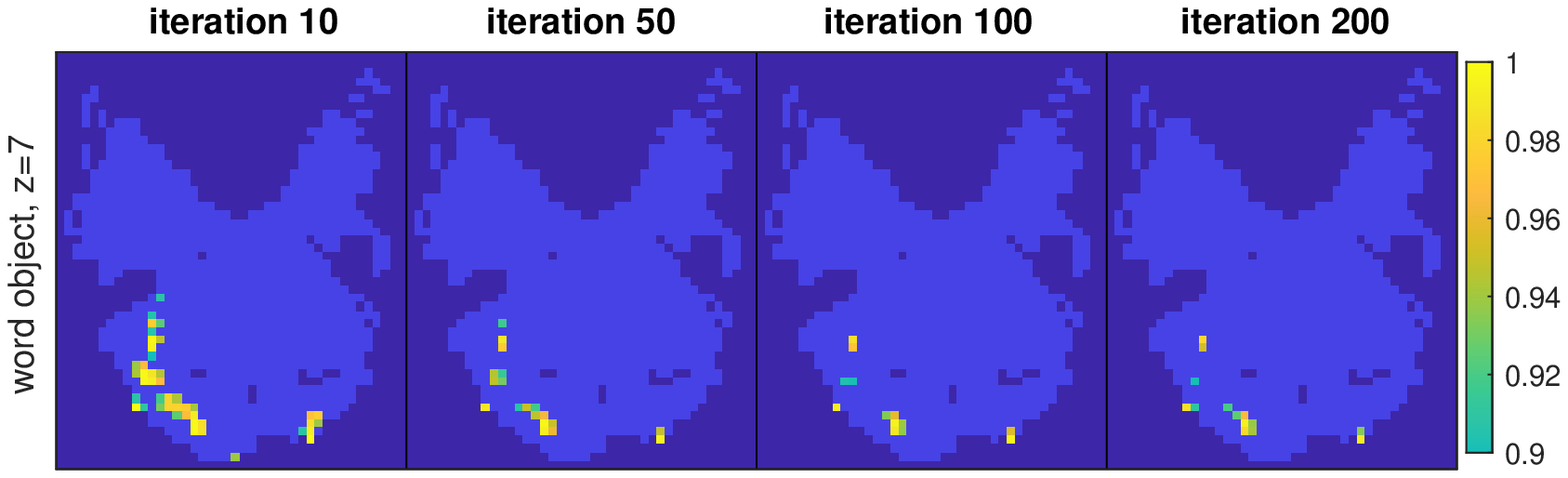}

}

\caption{Convergence of the EB method for the M$(2)$ prior. (a) The hyperparameters $\tau_{k}^{2}$ and $\kappa_{k}^{2}$ corresponding to different conditions over the iterations of the Algorithm~\ref{Alg:SGD}, when using the M$(2)$ prior for the face repetition data (left) and word object data (right). (b) PPM for the word object data after 10, 50, 100 and 200 iterations, where the last is the same as in Fig.~\ref{fig:realPPMs}.}
\label{fig:Convergence}
\end{figure}

\begin{figure}
\includegraphics[width=0.9\linewidth,trim={1mm 6mm 4mm 8mm},clip]{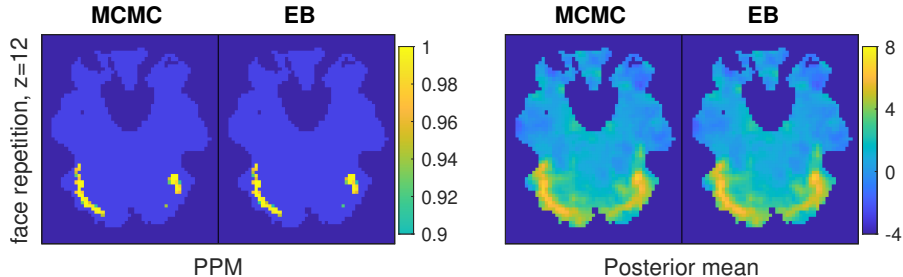}

\caption{Comparison between MCMC and EB in terms of PPMs (left) and posterior mean of activity coefficients (right) using the ICAR$(1)$ prior for the face repetition data. The presented PPM for EB is the same as in Fig.~\ref{fig:realPPMs}.}
\label{fig:MCMCvsEBPPMs}
\end{figure}

To assess how well the EB posterior with optimized hyperparameters approximates the full posterior, we also fit the model with the ICAR(1) prior using MCMC as described in \citet{Sid??n2017}, using the face repetition data. Fig.~\ref{fig:MCMCvsEBPPMs} compares PPMs and posterior mean maps between the two methods, and the differences are practically negligible, and much smaller than, for example, the differences between different spatial priors. The EB estimates for the spatial hyperparameters $\left\{ \tau_{k}^{2}\right\} $ are also very similar to the MCMC posterior mean. For this exercise the same conjugate gamma prior for $\tau_{k}^{2}$ as in \citet{Sid??n2017} was also for EB. The MCMC method used 10,000 iterations after 1,000 burnin samples and thinning factor 5.

These results support the conjecture made earlier, that the posterior distributions of the hyperparameters $\boldsymbol{\theta}$ are well approximated by point masses when the goal of the analysis is to correctly model the distribution of the activity coefficients $\mathbf{W}$. It would be interesting to do the same comparison for the other spatial priors, and the other hyperparameters ($\kappa
^2$, $h_x$ and $h_y$), but to our knowledge there exists no computationally feasible MCMC method to sample these parameters, which lack the conjugacy exploited for $\tau^2$. 



\section{Conclusions and directions for future research\label{sec:Discussion-and-future}}

We propose an efficient Bayesian inference algorithm for whole-brain analysis of fMRI data using the flexible and interpretable Mat\'{e}rn class of spatial priors. We study the empirical properties of the prior on simulated and two real fMRI datasets and conclude that the second order Matérn priors (M$(2)$ or A-M$(2)$) should be the preferred choice for future studies. The priors with $\alpha=1$ are clearly inferior in the sense that they do not find the seemingly correct activity patterns that are found by the priors with $\alpha=2$, they produce new samples that appear too speckled and they perform worse in the cross-validation. The differences between the M$(2)$ and ICAR$(2)$ are less evident, but fact that
they produce somewhat different activity maps for some datasets, that
new samples from the ICAR$(2)$ look too smooth, that M$(2)$ performed
consistently better in the cross-validation, and that the M$(2)$
prior parameters are easier interpreted all argues in favor of
the M$(2)$ prior. 

The introduced anisotropic Matérn prior was shown to perform slightly better than the isotropic Matérn prior in the cross-validation, but overall the differences between the results for the two priors is quite small. Still, A-M$(2)$ has the capacity to model also anisotropic datasets, while containing the M$(2)$ prior as a special case, and could therefore be the best alternative.

The optimization algorithm appears satisfactory with relatively fast
convergence. Using SGD is an improvement relative to the coordinate
descent algorithm employed for SVB in \citet{Sid??n2017}, because
following the gradient is in general the shorter way to reach the
optimum and there exists better theoretical guarantees for the convergence.
Also, well-known acceleration strategies, such as using momentum or
the approximate Hessian information, are easier to adopt to SGD and
one can thereby avoid the more ad hoc acceleration strategies used
in \citet[Appendix C]{Sid??n2017}.

The EB method is shown to approximate the exact MCMC posterior well empirically, suggesting that properly accounting for the uncertainty in the spatial hyperparameters is of minor importance if the main object is the activity maps.

As the smoothness parameter $\alpha$ appears to be the most important for the resulting activity maps, it would in future research be interesting to estimate it as a non-integer value, which could be addressed using the method in \citet{bolin2020rational}.

The PPMs reported in this work only contain the marginal probability of activation in each voxel. If instead using joint PPMs \citep{Yue2014a,Mejia2019} based on excursions sets \citep{Bolin2014a} to address the multiple comparison problem of classifying active voxels, it is likely to see larger differences between the M(2) and ICAR(2) prior, since the joint PPMs depend more on the spatial correlation. The joint PPMs are easily computed from MCMC output, but harder for the EB method due to the posterior covariance matrix being costly to compute. Future work should address this issue, which could probably be solved by extending the block RBMC method in \citet{Siden2018}.


The estimated spatial hyperparameters for the real datasets in Table~\ref{tab:Hyperparameters}
have strikingly similar values across different HRF regressors. A
natural idea is therefore to let these regressors share the same spatial
hyperparameters, at least when the tasks in the experiment are similar. Our Bayesian
inference algorithm is straightforwardly extended to this setting.


\bibliographystyle{apalike} 
\bibliography{Bibtex/library,Bibtex/manually} 

\begin{thebibliography}{}

\bibitem[Asmussen and Glynn, 2007]{asmussen2007stochastic}
Asmussen, S. and Glynn, P.~W. (2007).
\newblock Stochastic simulation: algorithms and analysis.
\newblock In {\em Stochastic modelling and applied probability}. Springer, New
  York.

\bibitem[Bezener et~al., 2018]{Bezener2018}
Bezener, M., Hughes, J., and Jones, G. (2018).
\newblock {Bayesian spatiotemporal modeling using hierarchical spatial priors
  with applications to functional magnetic resonance imaging}.
\newblock {\em Bayesian Analysis}, 13(4):1261--1313.

\bibitem[Bishop, 2006]{Bishop2006}
Bishop, C.~M. (2006).
\newblock {\em {Pattern recognition and machine learning}}.
\newblock Springer, New York.

\bibitem[Bolin and Kirchner, 2020]{bolin2020rational}
Bolin, D. and Kirchner, K. (2020).
\newblock The rational {SPDE} approach for {G}aussian random fields with
  general smoothness.
\newblock {\em Journal of Computational and Graphical Statistics},
  29(2):274--285.

\bibitem[Bolin and Lindgren, 2015]{Bolin2014a}
Bolin, D. and Lindgren, F. (2015).
\newblock {Excursion and contour uncertainty regions for latent Gaussian
  models}.
\newblock {\em Journal of the Royal Statistical Society: Series B (Statistical
  Methodology)}, 77(1):85--106.

\bibitem[Bolin et~al., 2019]{Bolin2018}
Bolin, D., Wallin, J., and Lindgren, F. (2019).
\newblock {Latent Gaussian random field mixture models}.
\newblock {\em Computational Statistics and Data Analysis}, 130:80--93.

\bibitem[Cox et~al., 2017]{Cox2017}
Cox, R.~W., Chen, G., Glen, D.~R., Reynolds, R.~C., and Taylor, P.~A. (2017).
\newblock {FMRI Clustering in AFNI: False-Positive Rates Redux}.
\newblock {\em Brain Connectivity}, 7(3):152--171.

\bibitem[Cryer and Chan, 2008]{Cryer2008}
Cryer, J.~D. and Chan, K.-S. (2008).
\newblock {\em Time series analysis with applications in R}.
\newblock Springer, New York, NY, second edition.

\bibitem[Duncan et~al., 2009]{Duncan2009}
Duncan, K., Pattamadilok, C., Knierim, I., and Devlin, J. (2009).
\newblock Consistency and variability in functional localisers.
\newblock {\em Neuroimage}, 46(4):1018--1026.

\bibitem[Eklund et~al., 2016]{Eklund2016}
Eklund, A., Nichols, T.~E., and Knutsson, H. (2016).
\newblock {Cluster failure: why fMRI inferences for spatial extent have
  inflated false positive rates}.
\newblock {\em Proceedings of the National Academy of Sciences},
  113(28):7900--7905.

\bibitem[Friston et~al., 1995]{Friston1995a}
Friston, K.~J., Holmes, a.~P., Worsley, K.~J., Poline, J.-P., Frith, C.~D., and
  Frackowiak, R. S.~J. (1995).
\newblock {Statistical parametric maps in functional imaging: A general linear
  approach}.
\newblock {\em Human Brain Mapping}, 2(4):189--210.

\bibitem[Fuglstad et~al., 2019]{Fuglstad2018}
Fuglstad, G.-A., Simpson, D., Lindgren, F., and Rue, H. (2019).
\newblock Constructing priors that penalize the complexity of {G}aussian random
  fields.
\newblock {\em Journal of the American Statistical Association},
  114(525):445--452.

\bibitem[Gneiting and Raftery, 2007]{Gneiting2007}
Gneiting, T. and Raftery, A.~E. (2007).
\newblock {Strictly proper scoring rules, prediction, and estimation}.
\newblock {\em Journal of the American Statistical Association},
  102(477):359--378.

\bibitem[G{\"{o}}ssl et~al., 2001]{Gossl2001}
G{\"{o}}ssl, C., Auer, D.~P., and Fahrmeir, L. (2001).
\newblock {Bayesian spatiotemporal inference in functional magnetic resonance
  imaging.}
\newblock {\em Biometrics}, 57(2):554--562.

\bibitem[Groves et~al., 2009]{Groves2009}
Groves, A.~R., Chappell, M.~A., and Woolrich, M.~W. (2009).
\newblock {Combined spatial and non-spatial prior for inference on MRI
  time-series}.
\newblock {\em NeuroImage}, 45(3):795--809.

\bibitem[Handcock and Stein, 1993]{Handcock1993}
Handcock, M.~S. and Stein, M.~L. (1993).
\newblock {A Bayesian analysis of kriging}.
\newblock {\em Technometrics}, 35(4):403--410.

\bibitem[Harrison and Green, 2010]{Harrison2010}
Harrison, L.~M. and Green, G. G.~R. (2010).
\newblock {A Bayesian spatiotemporal model for very large data sets}.
\newblock {\em NeuroImage}, 50(3):1126--1141.

\bibitem[Henson et~al., 2002]{Henson2002}
Henson, R., Shallice, T., Gorno-Tempini, M.~L., and Dolan, R. (2002).
\newblock {Face repetition effects in implicit and explicit memory tests as
  measured by fMRI}.
\newblock {\em Cerebral Cortex}, 12:178--186.

\bibitem[Hutchinson, 1990]{Hutchinson1990}
Hutchinson, M.~F. (1990).
\newblock {A stochastic estimator of the trace of the influence matrix for
  Laplacian smoothing splines}.
\newblock {\em Communications in Statistics-Simulation and Computation},
  19(2):433--450.

\bibitem[Lange, 1995]{Lange1995}
Lange, K. (1995).
\newblock {A gradient algorithm locally equivalent to the EM algorithm}.
\newblock {\em Journal of the Royal Statistical Society, Series B},
  57(2):425--437.

\bibitem[Lee et~al., 2014]{Lee2014}
Lee, K.-J., Jones, G.~L., Caffo, B.~S., and Bassett, S.~S. (2014).
\newblock {Spatial Bayesian variable selection models on functional magnetic
  resonance imaging time-series data}.
\newblock {\em Bayesian Analysis}, 9(3):699--732.

\bibitem[Lindgren et~al., 2011]{Lindgren2011}
Lindgren, F., Rue, H., and Lindstr{\"{o}}m, J. (2011).
\newblock {An explicit link between Gaussian fields and Gaussian Markov random
  fields: The SPDE approach}.
\newblock {\em Journal of the Royal Statistical Society Series B},
  73(4):423--498.

\bibitem[Lindquist, 2008]{Lindquist2008}
Lindquist, M. (2008).
\newblock {The statistical analysis of fMRI data}.
\newblock {\em Statistical Science}, 23(4):439--464.

\bibitem[Mat{\'{e}}rn, 1960]{Matern1960}
Mat{\'{e}}rn, B. (1960).
\newblock {\em {Spatial variation}}.
\newblock PhD thesis.

\bibitem[Mejia et~al., 2020]{Mejia2019}
Mejia, A.~F., Yue, Y.~R., Bolin, D., Lindgren, F., and Lindquist, M.~A. (2020).
\newblock A {B}ayesian general linear modeling approach to cortical surface
  f{M}{R}{I} data analysis.
\newblock {\em Journal of the American Statistical Association},
  115(530):501--520.

\bibitem[Papandreou and Yuille, 2010]{Papandreou2010}
Papandreou, G. and Yuille, A. (2010).
\newblock {Gaussian sampling by local perturbations}.
\newblock {\em Advances in Neural Information Processing Systems 23},
  90(8):1858--1866.

\bibitem[Penny et~al., 2007]{Penny2007}
Penny, W.~D., Flandin, G., and Trujillo-Barreto, N.~J. (2007).
\newblock {Bayesian comparison of spatially regularised general linear models}.
\newblock {\em Human Brain Mapping}, 28(4):275--293.

\bibitem[Penny et~al., 2005]{pennyEtAlSpatialPrior2005}
Penny, W.~D., Trujillo-Barreto, N.~J., and Friston, K.~J. (2005).
\newblock {Bayesian fMRI time series analysis with spatial priors}.
\newblock {\em NeuroImage}, 24(2):350--362.

\bibitem[Petersen and Pedersen, 2012]{Petersen2007}
Petersen, K.~B. and Pedersen, M.~S. (2012).
\newblock {\em {The matrix cookbook}}.
\newblock Version 20121115.

\bibitem[Poldrack and Gorgolewski, 2017]{poldrack2017openfmri}
Poldrack, R.~A. and Gorgolewski, K.~J. (2017).
\newblock Openf{MRI}: Open sharing of task f{MRI} data.
\newblock {\em Neuroimage}, 144:259--261.

\bibitem[Rasmussen and Williams, 2006]{Seeger2004}
Rasmussen, C.~E. and Williams, K.~I. (2006).
\newblock {\em {Gaussian processes for machine learning}}.
\newblock MIT Press.

\bibitem[Robbins and Monro, 1951]{RobbinsMonro}
Robbins, H. and Monro, S. (1951).
\newblock A stochastic approximation method.
\newblock {\em The Annals of Mathematical Statistics}, 22(3):400--407.

\bibitem[Rue and Held, 2005]{Isham2004}
Rue, H. and Held, L. (2005).
\newblock {\em {Gaussian Markov random fields: theory and applications}}.
\newblock CRC Press.

\bibitem[Rue and Martino, 2007]{Rue2007}
Rue, H. and Martino, S. (2007).
\newblock {Approximate Bayesian inference for hierarchical Gaussian Markov
  random field models}.
\newblock {\em Journal of Statistical Planning and Inference},
  137(10):3177--3192.

\bibitem[Rue et~al., 2009]{Solis-Trapala2009}
Rue, H., Martino, S., and Chopin, N. (2009).
\newblock {Approximate Bayesian inference for latent Gaussian models by using
  integrated nested Laplace approximation}.
\newblock {\em Journal of the Royal Statistical Society, Series B},
  71(2):319--392.

\bibitem[Sid{\'{e}}n et~al., 2017]{Sid??n2017}
Sid{\'{e}}n, P., Eklund, A., Bolin, D., and Villani, M. (2017).
\newblock {Fast Bayesian whole-brain fMRI analysis with spatial 3D priors}.
\newblock {\em NeuroImage}, 146:211--225.

\bibitem[Sid{\'{e}}n et~al., 2018]{Siden2018}
Sid{\'{e}}n, P., Lindgren, F., Bolin, D., and Villani, M. (2018).
\newblock {Efficient covariance approximations for large sparse precision
  matrices}.
\newblock {\em Journal of Computational and Graphical Statistics},
  27(4):898--909.

\bibitem[Simpson et~al., 2017]{Simpson2017}
Simpson, D.~P., Rue, H., Riebler, A., Martins, T.~G., and S{\o}rbye, S.~H.
  (2017).
\newblock {Penalising model component complexity: A principled, practical
  approach to constructing priors}.
\newblock {\em Statistical Science}, 32(1):1--28.

\bibitem[Smith and Fahrmeir, 2007]{Smith2007}
Smith, M. and Fahrmeir, L. (2007).
\newblock {Spatial Bayesian variable selection with application to functional
  magnetic resonance imaging}.
\newblock {\em Journal of the American Statistical Association}, 102:417--431.

\bibitem[Stein, 1999]{Stein1999}
Stein, M.~L. (1999).
\newblock {\em {Interpolation of spatial data. Some theory for kriging}}.
\newblock Springer-Verlag, New York.

\bibitem[Takahashi et~al., 1973]{Takahashi1973}
Takahashi, K., Fagan, J., and Chen, M.~S. (1973).
\newblock {Formation of a sparse bus impedance matrix and its application to
  short circuit study}.
\newblock {\em IEEE Power Industry Computer Applications Conference}, pages
  63--69.

\bibitem[Vincent et~al., 2010]{Vincent2010a}
Vincent, T., Risser, L., and Ciuciu, P. (2010).
\newblock {Spatially adaptive mixture modeling for analysis of fMRI time
  series}.
\newblock {\em IEEE transactions on medical imaging}, 29(4):1059--1074.

\bibitem[Whittle, 1954]{whittle1954}
Whittle, P. (1954).
\newblock On stationary processes in the plane.
\newblock {\em Biometrika}, 41:434--449.

\bibitem[Whittle, 1963]{whittle1963}
Whittle, P. (1963).
\newblock Stochastic processes in several dimensions.
\newblock {\em Bulletin of the International Statistical Institute},
  40(2):974--994.

\bibitem[Woolrich et~al., 2004]{Woolrich2004c}
Woolrich, M.~W., Jenkinson, M., Brady, J.~M., and Smith, S.~M. (2004).
\newblock {Fully Bayesian spatio-temporal modeling of fMRI data.}
\newblock {\em IEEE transactions on medical imaging}, 23(2):213--31.

\bibitem[Yue et~al., 2014]{Yue2014a}
Yue, Y.~R., Lindquist, M., Bolin, D., Lindgren, F., Simpson, D., and Rue, H.
  (2014).
\newblock {A Bayesian general linear modeling approach to slice-wise fMRI data
  analysis}.
\newblock {\em Preprint}.

\bibitem[Zhang et~al., 2014]{Zhang2014}
Zhang, L., Guindani, M., Versace, F., and Vannucci, M. (2014).
\newblock {A spatio-temporal nonparametric Bayesian variable selection model of
  fMRI data for clustering correlated time courses}.
\newblock {\em NeuroImage}, 95:162--175.

\end{thebibliography}

\begin{acknowledgement}
This work was funded by Swedish Research Council (Vetenskapsr\aa det)
grant no 2013-5229 and grant no 2016-04187. Finn Lindgren was funded
by the European Union's Horizon 2020 Programme for Research and Innovation,
no 640171, EUSTACE. Anders Eklund was funded by Center for Industrial
Information Technology (CENIIT) at Link\"{o}ping University.
\end{acknowledgement}

\newpage

\begin{supplement}

\section{Derivation of the gradient and approximate Hessian\label{Appendix gradient}}

For the optimization of the parameters $\boldsymbol{\theta}=\left\{ \boldsymbol{\theta}_{s},\boldsymbol{\lambda},\mathbf{A}\right\} $,
we need the gradient and approximate Hessian with respect to the different
hyperparameters of the log marginal likelihood $\log p\left(\mathbf{y}|\boldsymbol{\theta}\right)=\log p\left(\mathbf{y}|\boldsymbol{\beta},\boldsymbol{\theta}\right)+\log p\left(\boldsymbol{\beta}|\boldsymbol{\theta}\right)-\log\left(\boldsymbol{\beta}|\mathbf{y},\boldsymbol{\theta}\right)$,
which is constant with respect to $\boldsymbol{\beta}$. The gradient
of the log posterior is then simply $\frac{\partial}{\partial\theta_{i}}\log p\left(\boldsymbol{\theta}|\mathbf{y}\right)=\frac{\partial}{\partial\theta_{i}}\log p\left(\mathbf{y}|\boldsymbol{\theta}\right)+\frac{\partial}{\partial\theta_{i}}\log p\left(\boldsymbol{\theta}\right)$,
and similarly for the approximate Hessian. We will start this derivation
by considering the log likelihood $\log p\left(\mathbf{y}|\boldsymbol{\beta},\boldsymbol{\theta}\right)$,
then the conditional log posterior $\log\left(\boldsymbol{\beta}|\mathbf{y},\boldsymbol{\theta}\right)$,
before producing the expressions for the log marginal likelihood gradient
as well as the approximate log marginal likelihood Hessian. Finally,
the corresponding posterior gradient and Hessian can be computed by
adding the prior contributions derived in the last subsection.

To get a more robust optimization algorithm in practice, we use the
reparameterizations $\tau_{0,k}=\log\left(\tau_{k}^{2}\right)$, $\kappa_{0,k}=\log\left(\kappa_{k}^{2}\right)$
and $\lambda_{0,n}=\log\left(\lambda_{n}\right)$, so that the parameters
are defined over the whole $\mathbb{R}$ and then perform the optimization
over these new variables. The gradient for the new variables is easily
obtained from the gradient for the old ones using the chain rule,
for example $\frac{\partial\log p\left(\mathbf{y}|\boldsymbol{\theta}\right)}{\partial\tau_{0,k}}=\frac{\partial\log p\left(\mathbf{y}|\boldsymbol{\theta}\right)}{\partial\tau_{k}^{2}}\cdot\tau_{k}^{2}$.
For $A_{p,n}$ we use the logit reparameterization $A_{0,p,n}=\log\left(\frac{1+A_{p,n}}{2}\right)-\log\left(\frac{1-A_{p,n}}{2}\right)$
, guaranteeing $A_{p,n}\in\left(-1,1\right)$, which is the stability
region for AR$\left(1\right)$, and we have $\frac{\partial\log p\left(\mathbf{y}|\boldsymbol{\theta}\right)}{\partial A_{0,p,n}}=\frac{\partial\log p\left(\mathbf{y}|\boldsymbol{\theta}\right)}{\partial A_{p,n}}\cdot\frac{1-A_{p,n}^{2}}{2}$.

\subsection{Log likelihood}
As in \citet[Appendix A]{Sid??n2017}, the log likelihood can be written
as
\begin{eqnarray}
\log p\left(\mathbf{y}|\boldsymbol{\beta},\boldsymbol{\theta}\right) & = & \frac{T-P}{2}\sum_{n=1}^{N}\log\left(\lambda_{n}\right)-\frac{1}{2}\sum_{n=1}^{N}\lambda_{n}l_{n}\left(\mathbf{W}_{\cdot,n}\right)+\text{const},\label{eq:loglike_iid}
\end{eqnarray}
where
\[
l_{n}\left(\mathbf{W}_{\cdot,n}\right)=\mathbf{Y}_{\cdot,n}^{T}\mathbf{Y}_{\cdot,n}-2\mathbf{Y}_{\cdot,n}^{T}\mathbf{X}\mathbf{W}_{\cdot,n}+\mathbf{W}_{\cdot,n}^{T}\mathbf{X}^{T}\mathbf{X}\mathbf{W}_{\cdot,n}
\]
in the case when the noise is independent over time and
\begin{align*}
l_{n}\left(\mathbf{W}_{\cdot,n}\right)=\, & \mathbf{Y}_{\cdot,n}^{T}\mathbf{Y}_{\cdot,n}-2\mathbf{Y}_{\cdot,n}^{T}\mathbf{X}\mathbf{W}_{\cdot,n}+\mathbf{W}_{\cdot,n}^{T}\mathbf{X}^{T}\mathbf{X}\mathbf{W}_{\cdot,n}-2\mathbf{Y}_{\cdot,n}^{T}\mathbf{d}_{n}^{T}\mathbf{A}_{\cdot,n}+\mathbf{A}_{\cdot,n}^{T}\mathbf{d}_{n}\mathbf{d}_{n}^{T}\mathbf{A}_{\cdot,n}\\
&+\mathbf{W}_{\cdot,n}^{T}\mathbf{B}_{n}^{T}\mathbf{A}_{\cdot,n}+\mathbf{A}_{\cdot,n}^{T}\mathbf{B}_{n}\mathbf{W}_{\cdot,n}-\mathbf{W}_{\cdot,n}^{T}\left(\mathbf{R}\mathbf{A}_{\cdot,n}+\left(\mathbf{R}\mathbf{A}_{\cdot,n}\right)^{T}\right)\mathbf{W}_{\cdot,n}\\
&-\mathbf{A}_{\cdot,n}^{T}\left(\mathbf{D}_{n}\mathbf{W}_{\cdot,n}+\left(\mathbf{D}_{n}\mathbf{W}_{\cdot,n}\right)^{T}\right)\mathbf{A}_{\cdot,n}+\mathbf{W}_{\cdot,n}^{T}\left(\mathbf{A}_{\cdot,n}^{T}\mathbf{S}\mathbf{A}_{\cdot,n}\right)\mathbf{W}_{\cdot,n}
\end{align*}
when the noise follows an AR$\left(P\right)$-process in each voxel
with AR-parameters $\mathbf{A}_{\cdot,n}$. We follow the notation
in \citet{Sid??n2017}, except here $\boldsymbol{\beta}=\text{vec}\left(\mathbf{W}^{T}\right)$.
For convenience we list the different matrices and tensors and their
sizes also here: 
\begin{align*}
\underset{P\times\left(T-P\right)}{\mathbf{d}_{n}}\text{ contains lagged values of }\mathbf{Y}_{\cdot,n},\,\,\,\,\, & \underset{P\times\left(T-P\right)\times K}{\tilde{\mathbf{X}}}\text{ contains lagged values of }\mathbf{X}\\
\underset{P\times K}{\mathbf{B}_{n}}=\mathbf{Y}_{\cdot,n}^{\prime}\tilde{\mathbf{X}}+\mathbf{d}_{n}\mathbf{X},\,\,\,\,\,\underset{K\times K\times P}{\mathbf{R}}=\mathbf{X}^{\prime}\tilde{\mathbf{X}},\,\,\,\,\, & \underset{P\times K\times P}{\mathbf{D}_{n}}=\mathbf{d}_{n}\tilde{\mathbf{X}}\,\,\,\,\,\underset{P\times K\times K\times P}{\mathbf{S}}=\tilde{\mathbf{X}}\tilde{\mathbf{X}}.
\end{align*}
The motivation for using this seemingly cumbersome notation is that
it allows for precomputation of all sums and matrix products over
the time dimension, so that these can be avoided in each iteration
of the algorithm, giving greatly reduced computation times.

\subsection{Conditional log posterior \label{subsec:Conditional-log-posterior}}

Following the derivation in \citet[Appendix A]{Sid??n2017}, we obtain \newline
$\boldsymbol{\beta}|\mathbf{y},\boldsymbol{\theta}\sim\mathcal{N}(\tilde{\boldsymbol{\mu}},\tilde{\mathbf{Q}}^{-1})=\mathcal{N}(\tilde{\mathbf{Q}}^{-1}\mathbf{b},\tilde{\mathbf{Q}}^{-1})$
with 
\begin{equation}
\tilde{\mathbf{Q}}=\mathbf{X}^{T}\mathbf{X}\otimes\text{diag}\left(\boldsymbol{\lambda}\right)+\mathbf{Q},\,\,\,\,\,\,\,\,\tilde{\boldsymbol{\mu}}=\tilde{\mathbf{Q}}^{-1}\text{vec}\left(\text{diag}\left(\boldsymbol{\lambda}\right)\mathbf{Y}^{T}\mathbf{X}\right),\label{eq:postbeta_iid}
\end{equation}
for $P=0$ and for the case with autoregressive noise $\left(P>0\right)$
we have 
\begin{align}
\tilde{\mathbf{Q}} & =\mathbf{P}_{KN}^{T}\underset{n\in\left\{ 1,\ldots,N\right\} }{\text{blkdiag}}\left[\lambda_{n}\tilde{\mathbf{Q}}_{n}\right]\mathbf{P}_{KN}+\mathbf{Q},\,\,\,\,\tilde{\boldsymbol{\mu}}=\tilde{\mathbf{Q}}^{-1}\text{vec}\left(\text{diag}\left(\boldsymbol{\lambda}\right)\left[\begin{array}{c}
\vdots\\
\tilde{\mathbf{q}}_{n}\\
\vdots
\end{array}\right]_{n\in\left\{ 1,\ldots,N\right\} }\right),\label{eq:postbeta_ar}\\
\tilde{\mathbf{Q}}_{n} & =\mathbf{X}^{T}\mathbf{X}-\mathbf{R}\mathbf{A}_{\cdot,n}-\left(\mathbf{R}\mathbf{A}_{\cdot,n}\right)^{T}+\mathbf{A}_{\cdot,n}^{T}\mathbf{S}\mathbf{A}_{\cdot,n},\,\,\,\,\,\,\tilde{\mathbf{q}}_{n}=\mathbf{Y}_{\cdot,n}^{T}\mathbf{X}-\mathbf{A}_{\cdot,n}^{T}\mathbf{B}_{n}+\mathbf{A}_{\cdot,n}^{T}\mathbf{D}_{n}\mathbf{A}_{\cdot,n}\nonumber 
\end{align}
 where $\mathbf{P}_{KN}$ is the permutation matrix defined such that
$\text{vec}\left(\mathbf{W}\right)=\mathbf{P}_{KN}\text{vec}\left(\mathbf{W}^{T}\right)$.
We note that we can also write the parameters of the i.i.d. case on
the second, slightly more complicated form in Eq. (\ref{eq:postbeta_ar})
by instead choosing $\tilde{\mathbf{Q}}_{n}=\mathbf{X}^{T}\mathbf{X}$
and $\tilde{\mathbf{q}}_{n}=\mathbf{Y}_{\cdot,n}^{T}\mathbf{X}$.

\subsection{Gradient\label{subsec:Gradient}}

In this section we derive the gradient for the M$(2)$ case. The gradient
for the other spatial priors can be obtained using the same strategy
and these will be left out for brevity. In summary, the log marginal
likelihood
\begin{align}
\log p\left(\mathbf{y}|\boldsymbol{\theta}\right)= & \,\frac{T-P}{2}\sum_{n=1}^{N}\log\left(\lambda_{n}\right)-\frac{1}{2}\sum_{n=1}^{N}\lambda_{n}l_{n}\left(\mathbf{W}_{\cdot,n}\right)+\frac{1}{2}\log\left|\mathbf{Q}\right|-\frac{1}{2}\boldsymbol{\beta}^{T}\mathbf{Q}\boldsymbol{\beta}\label{eq:margloglike}\\
 & -\frac{1}{2}\log\left|\tilde{\mathbf{Q}}\right|+\frac{1}{2}\left(\boldsymbol{\beta}-\tilde{\boldsymbol{\mu}}\right)^{T}\tilde{\mathbf{Q}}\left(\boldsymbol{\beta}-\tilde{\boldsymbol{\mu}}\right)+\text{const}.\nonumber 
\end{align}
We begin by writing down the log marginal likelihood gradient and
then the derivation follows.
\begin{align}
\frac{\partial\log p\left(\mathbf{y}|\boldsymbol{\theta}\right)}{\partial\tau_{k}^{2}} & =\frac{N}{2\tau_{k}^{2}}-\frac{1}{2}\text{tr}\left(\tilde{\mathbf{Q}}^{-1}\left(\mathbf{J}^{kk}\otimes\mathbf{K}_{k}\mathbf{K}_{k}\right)\right)-\frac{1}{2}\mathbf{M}_{k,\cdot}\mathbf{K}_{k}\mathbf{K}_{k}\mathbf{M}_{k,\cdot}^{T},\label{eq:gradient}\\
\frac{\partial\log p\left(\mathbf{y}|\boldsymbol{\theta}\right)}{\partial\kappa_{k}^{2}} & =\text{tr}\left(\mathbf{K}_{k}^{-1}\right)-\tau_{k}^{2}\text{tr}\left(\tilde{\mathbf{Q}}^{-1}\left(\mathbf{J}^{kk}\otimes\mathbf{K}_{k}\right)\right)-\tau_{k}^{2}\mathbf{M}_{k,\cdot}\mathbf{K}_{k}\mathbf{M}_{k,\cdot}^{T},\nonumber \\
\frac{\partial\log p\left(\mathbf{y}|\boldsymbol{\theta}\right)}{\partial\lambda_{n}} & =\frac{T-P}{2\lambda_{n}}-\frac{1}{2}\text{tr}\left(\tilde{\mathbf{Q}}^{-1}\mathbf{P}_{KN}^{T}\left(\mathbf{J}^{nn}\otimes\tilde{\mathbf{Q}}_{n}\right)\mathbf{P}_{KN}\right)-\frac{1}{2}l_{n}\left(\mathbf{M}_{\cdot,n}\right),\nonumber \\
\frac{\partial\log p\left(\mathbf{y}|\boldsymbol{\theta}\right)}{\partial\mathbf{A}_{\cdot,n}} & =-\frac{1}{2}\lambda_{n}\left.\frac{\partial l_{n}\left(\mathbf{W}_{\cdot,n}\right)}{\partial\mathbf{A}_{.,n}}\right|_{\mathbf{W}=\mathbf{M}}-\frac{1}{2}\frac{\partial\log\left|\tilde{\mathbf{Q}}\right|}{\partial\mathbf{A}_{\cdot,n}}\nonumber, 
\end{align}
where $\mathbf{K}_{k}=\mathbf{G}+\kappa_{k}^{2}\mathbf{I}$ so that
$\mathbf{Q}_{k}=\tau_{k}^{2}\mathbf{K}_{k}\mathbf{K}_{k}$ and $\mathbf{J}^{ij}$
is the square single-entry matrix which is zero everywhere except
in $\left(i,j\right)$ where it is $1$. The size of $\mathbf{J}^{ij}$
is clear from the context and is here used in couple with the Kronecker
product to construct single-block matrices, where everything but one
block is zero.\textbf{ M} is the $K\times N$ matrix such that $\tilde{\boldsymbol{\mu}}=\text{vec}\left(\mathbf{M}^{T}\right)$
(compare with $\boldsymbol{\beta}=\text{vec}\left(\mathbf{W}^{T}\right)$).
The terms in the expression for $\frac{\partial\log p\left(\mathbf{y}|\boldsymbol{\theta}\right)}{\partial\mathbf{A}_{\cdot,n}}$
are given in Eq. (\ref{eq:AR gradient terms}).

We begin with computing the gradient with respect to $\tau_{k}^{2}$
and $\kappa_{k}^{2}$, that do not appear in the likelihood, which
is why the gradient is the same for the case $P=0$ and $P>0$ (given
$\tilde{\mathbf{Q}}$ and $\tilde{\boldsymbol{\mu}}$). Thereafter,
we treat $\lambda_{n}$ which is different in these cases depending
on the log likelihood term expression $l_{n}$ and lastly $A_{pn}$
which is only relevant for the case $P>0$. A reference for some of
the matrix algebraic operations used here is \citet{Petersen2007}.

\subsubsection*{Gradient with respect to $\tau_{k}^{2}$ and $\kappa_{k}^{2}$}

Some useful derivatives are
\begin{align*}
\frac{\partial\mathbf{Q}_{k}}{\partial\tau_{k}^{2}} & =\mathbf{K}_{k}\mathbf{K}_{k},\,\,\,\,\,\,\frac{\partial\mathbf{Q}_{k}}{\partial\kappa_{k}^{2}}=2\tau_{k}^{2}\mathbf{K}_{k},\,\,\,\,\,\,\frac{\partial\mathbf{Q}}{\partial\tau_{k}^{2}}=\mathbf{J}^{kk}\otimes\mathbf{K}_{k}\mathbf{K}_{k},\,\,\,\,\,\,\frac{\partial\mathbf{Q}}{\partial\kappa_{k}^{2}}=\mathbf{J}^{kk}\otimes2\tau_{k}^{2}\mathbf{K}_{k}.
\end{align*}
Also note that $\frac{\partial\mathbf{\tilde{Q}}}{\partial\tau_{k}^{2}}=\frac{\partial\mathbf{Q}}{\partial\tau_{k}^{2}}$
and $\frac{\partial\mathbf{\tilde{Q}}}{\partial\kappa_{k}^{2}}=\frac{\partial\mathbf{Q}}{\partial\kappa_{k}^{2}}$.
Furthermore
\begin{align*}
\frac{\partial\log\left|\mathbf{Q}_{k}\right|}{\partial\tau_{k}^{2}} & =\frac{\left|\mathbf{Q}_{k}\right|}{\left|\mathbf{Q}_{k}\right|}\text{tr}\left(\mathbf{Q}_{k}^{-1}\frac{\partial\mathbf{Q}_{k}}{\partial\tau_{k}^{2}}\right)=\text{tr}\left(\frac{1}{\tau_{k}^{2}}\mathbf{I}_{N}\right)=\frac{N}{\tau_{k}^{2}},\,\,\,\,\,\,\frac{\partial\log\left|\mathbf{Q}_{k}\right|}{\partial\kappa_{k}^{2}}=2\text{tr}\left(\mathbf{K}_{k}^{-1}\right),\\
\frac{\partial\log\left|\tilde{\mathbf{Q}}\right|}{\partial\tau_{k}^{2}} & =\text{tr}\left(\tilde{\mathbf{Q}}^{-1}\left(\mathbf{J}^{kk}\otimes\mathbf{K}_{k}\mathbf{K}_{k}\right)\right),\,\,\,\,\,\,\frac{\partial\log\left|\tilde{\mathbf{Q}}\right|}{\partial\kappa_{k}^{2}}=\text{tr}\left(\tilde{\mathbf{Q}}^{-1}\left(\mathbf{J}^{kk}\otimes2\tau_{k}^{2}\mathbf{K}_{k}\right)\right),\\
\frac{\partial\log\tilde{\boldsymbol{\mu}}^{T}\tilde{\mathbf{Q}}\tilde{\boldsymbol{\mu}}}{\partial\tau_{k}^{2}} & =\mathbf{b}^{T}\frac{\partial\tilde{\mathbf{Q}}^{-1}}{\partial\tau_{k}^{2}}\mathbf{b}=-\mathbf{b}^{T}\tilde{\mathbf{Q}}^{-1}\frac{\partial\tilde{\mathbf{Q}}}{\partial\tau_{k}^{2}}\tilde{\mathbf{Q}}^{-1}\mathbf{b}=-\tilde{\boldsymbol{\mu}}^{T}\left(\mathbf{J}^{kk}\otimes\mathbf{K}_{k}\mathbf{K}_{k}\right)\tilde{\boldsymbol{\mu}}\\ &=-\mathbf{M}_{k,\cdot}\mathbf{K}_{k}\mathbf{K}_{k}\mathbf{M}_{k,\cdot}^{T},\\
\frac{\partial\log\tilde{\boldsymbol{\mu}}^{T}\tilde{\mathbf{Q}}\tilde{\boldsymbol{\mu}}}{\partial\kappa_{k}^{2}} & =-\tilde{\boldsymbol{\mu}}^{T}\left(\mathbf{J}^{kk}\otimes2\tau_{k}^{2}\mathbf{K}_{k}\right)\tilde{\boldsymbol{\mu}}=-2\tau_{k}^{2}\mathbf{M}_{k,\cdot}\mathbf{K}_{k}\mathbf{M}_{k,\cdot}^{T}.
\end{align*}
Now, by setting $\boldsymbol{\beta}=\mathbf{0}$ and removing everything
that is constant with respect to $\boldsymbol{\tau}^{2}$ and $\boldsymbol{\kappa^{2}}$,
Eq. (\ref{eq:margloglike}) becomes
\[
\log p\left(\mathbf{y}|\boldsymbol{\theta}\right)=\frac{1}{2}\sum_{k=1}^{K}\log\left|\mathbf{Q}_{k}\right|-\frac{1}{2}\log\left|\tilde{\mathbf{Q}}\right|+\frac{1}{2}\tilde{\boldsymbol{\mu}}^{T}\tilde{\mathbf{Q}}\tilde{\boldsymbol{\mu}}+\text{const},
\]
and we get the derivatives with respect to $\tau_{k}^{2}$ and $\kappa_{k}^{2}$
in Eq. (\ref{eq:gradient}).

\subsubsection*{Gradient with respect to $\lambda_{n}$}

Note that
\begin{align*}
\frac{\partial\mathbf{\tilde{Q}}}{\partial\lambda_{n}^{2}} & =\mathbf{P}_{KN}^{T}\left(\mathbf{J}^{nn}\otimes\tilde{\mathbf{Q}}_{n}\right)\mathbf{P}_{KN}\text{,\,\,\,\,\,\,\,}\frac{\partial\log\left|\tilde{\mathbf{Q}}\right|}{\partial\lambda_{n}^{2}}=\text{tr}\left(\tilde{\mathbf{Q}}^{-1}\mathbf{P}_{KN}^{T}\left(\mathbf{J}^{nn}\otimes\tilde{\mathbf{Q}}_{n}\right)\mathbf{P}_{KN}\right),\\
 & \frac{\partial\log\left(\boldsymbol{\beta}-\tilde{\boldsymbol{\mu}}\right)^{T}\tilde{\mathbf{Q}}\left(\boldsymbol{\beta}-\tilde{\boldsymbol{\mu}}\right)}{\partial\lambda_{n}^{2}}=\left(\boldsymbol{\beta}-\tilde{\boldsymbol{\mu}}\right)^{T}\frac{\partial\mathbf{\tilde{Q}}}{\partial\lambda_{n}^{2}}\left(\boldsymbol{\beta}-\tilde{\boldsymbol{\mu}}\right)-2\frac{\partial\tilde{\boldsymbol{\mu}}^{T}}{\partial\lambda_{n}^{2}}\tilde{\mathbf{Q}}\left(\boldsymbol{\beta}-\tilde{\boldsymbol{\mu}}\right),
\end{align*}

and that the last expression is zero for $\boldsymbol{\beta}=\tilde{\boldsymbol{\mu}}$.
Thus, it is clear that the gradient with respect to $\lambda_{n}$
in Eq. (\ref{eq:gradient}) can be obtained from taking the derivative
of Eq. (\ref{eq:margloglike}) and evaluating at $\boldsymbol{\beta}=\tilde{\boldsymbol{\mu}}$.

\subsubsection*{Gradient with respect to $A_{pn}$}

Use that
\begin{align}
\frac{\partial l_{n}\left(\mathbf{W}_{\cdot,n}\right)}{\partial\mathbf{A}_{.,n}}= & 2\left[-\mathbf{Y}_{\cdot,n}^{T}\mathbf{d}_{n}^{T}+\mathbf{W}\mathbf{B}_{n}^{T}-\mathbf{W}_{\cdot,n}^{T}\mathbf{R}\mathbf{W}_{\cdot,n}\right.\label{eq:AR gradient terms}\\
 & \left.+\left(\mathbf{d}_{n}\mathbf{d}_{n}^{T}-\mathbf{D}_{n}\mathbf{W}_{\cdot,n}-\left(\mathbf{D}_{n}\mathbf{W}_{\cdot,n}\right)^{T}+\mathbf{W}_{\cdot,n}^{T}\mathbf{S}\mathbf{W}_{\cdot,n}\right)\mathbf{A}_{\cdot,n}\right],\nonumber \\
\frac{\partial\mathbf{\tilde{Q}}_{n}}{\partial A_{p,n}}= & -\mathbf{R}_{p}-\mathbf{R}_{p}^{T}+\mathbf{S}_{p}\mathbf{A}_{\cdot,n}+\left(\mathbf{S}_{p}\mathbf{A}_{\cdot,n}\right)^{T},\nonumber \\
\frac{\partial\log\left|\tilde{\mathbf{Q}}\right|}{\partial A_{p,n}}= & \text{tr}\left(\tilde{\mathbf{Q}}^{-1}\mathbf{P}_{KN}^{T}\left(\mathbf{J}^{nn}\otimes\lambda_{n}\frac{\partial\mathbf{\tilde{Q}}_{n}}{\partial A_{p,n}}\right)\mathbf{P}_{KN}\right),\nonumber 
\end{align}
where $\mathbf{R}_{p}$ and $\mathbf{S}_{p}$ are of sizes $K\times K$
and $K\times K\times P$ and refers to the $p$th sub-tensor from
the appropriate dimension of $\mathbf{R}$ and $\mathbf{S}$ respectively.
The expression in Eq. (\ref{eq:margloglike}) is derived after noting
that the remaining terms of the log likelihood become zero after taking
the derivative and evaluating at $\boldsymbol{\beta}=\tilde{\boldsymbol{\mu}}$.

\subsection{Approximate Hessian\label{subsec:Approximate-Hessian}}

The approximate Hessian for the log marginal likelihood in the M$\left(2\right)$
case, computed directly with respect to the parameters $\tau_{0,k}=\log\left(\tau_{k}^{2}\right)$,
$\kappa_{0,k}=\log\left(\kappa_{k}^{2}\right)$ is
\begin{align}
E_{\boldsymbol{\beta}|\mathbf{Y},\boldsymbol{\theta}}\left[\frac{\partial^{2}\log p\left(\mathbf{y},\boldsymbol{\beta}|\boldsymbol{\theta}\right)}{\partial\tau_{0,k}^{2}}\right] & =-\frac{\tau_{k}^{2}}{2}\text{tr}\left(\tilde{\mathbf{Q}}^{-1}\left(\mathbf{J}^{kk}\otimes\mathbf{K}_{k}\mathbf{K}_{k}\right)\right)-\frac{1}{2}\mathbf{M}_{k,\cdot}\mathbf{Q}_{k}\mathbf{M}_{k,\cdot}^{T},\label{eq:Hessian}\\
E_{\boldsymbol{\beta}|\mathbf{Y},\boldsymbol{\theta}}\left[\frac{\partial^{2}\log p\left(\mathbf{y},\boldsymbol{\beta}|\boldsymbol{\theta}\right)}{\partial\kappa_{0,k}^{2}}\right] & =\kappa_{k}^{2}\left[\text{tr}\left(\mathbf{K}_{k}^{-1}\right)-\kappa_{k}^{2}\tau_{k}^{2}\text{tr}\left(\mathbf{K}_{k}^{-1}\mathbf{K}_{k}^{-1}\right)+\right.\nonumber \\
 & \,\,\,\,\,\,-\tau_{k}^{2}\mathbf{M}_{k,\cdot}\mathbf{K}_{k}\mathbf{M}_{k,\cdot}^{T}-\tau_{k}^{2}\text{tr}\left(\tilde{\mathbf{Q}}^{-1}\left(\mathbf{J}^{kk}\otimes\mathbf{K}_{k}\right)\right)+\\ & \,\,\,\,\,\,-\kappa_{k}^{2}\tau_{k}^{2}\mathbf{M}_{k,\cdot}\mathbf{M}_{k,\cdot}^{T}-\kappa_{k}^{2}\tau_{k}^{2}\text{tr}\left(\tilde{\mathbf{Q}}^{-1}\left(\mathbf{J}^{kk}\otimes\mathbf{I}\right)\right).\nonumber 
\end{align}
The derivation starts by noting that the non-constant part of the
augmented log likelihood with respect to $\tau_{0,k}$ and $\kappa_{0,k}$
is
\begin{equation}
\log p\left(\mathbf{y},\boldsymbol{\beta}|\boldsymbol{\theta}\right)=\,\frac{1}{2}\log\left|\mathbf{Q}\right|-\frac{1}{2}\boldsymbol{\beta}^{T}\mathbf{Q}\boldsymbol{\beta}+\text{const}.\label{eq:augmented-likelihood}
\end{equation}
Taking the derivative twice with respect to $\tau_{0,k}$ and $\kappa_{0,k}$
and computing the expectation gives the result, after noting that
$E_{\boldsymbol{\beta}|\mathbf{Y},\boldsymbol{\theta}}\left[\boldsymbol{\beta}^{T}\mathbf{T}\boldsymbol{\beta}\right]=\text{tr}\left(\tilde{\mathbf{Q}}^{-1}\mathbf{T}\right)+\tilde{\boldsymbol{\mu}}^{T}\mathbf{T}\tilde{\boldsymbol{\mu}}$,
for general $KN\times KN$ matrix $\mathbf{T}$ \citep[Eq. (318)]{Petersen2007}.
We write out the derivation for $\tau_{0,k}$ as an example. First
note that
\[
\frac{\partial\mathbf{Q}}{\partial\tau_{0,k}}=\mathbf{J}^{kk}\otimes\frac{\partial}{\partial\tau_{0,k}}\exp\left(\tau_{0,k}\right)\mathbf{K}_{k}\mathbf{K}_{k}=\mathbf{J}^{kk}\otimes\mathbf{Q}_{k},
\]
so
\begin{align*}
\frac{\partial\log p\left(\mathbf{y},\boldsymbol{\beta}|\boldsymbol{\theta}\right)}{\partial\tau_{0,k}} & =\frac{1}{2}\text{tr}\left(\mathbf{Q}_{k}^{-1}\mathbf{Q}_{k}\right)-\frac{1}{2}\mathbf{W}_{k,\cdot}\mathbf{Q}_{k}\mathbf{W}_{k,\cdot}^{T}=\frac{N}{2}-\frac{1}{2}\mathbf{W}_{k,\cdot}\mathbf{Q}_{k}\mathbf{W}_{k,\cdot}^{T}\\
\frac{\partial^{2}\log p\left(\mathbf{y},\boldsymbol{\beta}|\boldsymbol{\theta}\right)}{\partial\tau_{0,k}^{2}} & =-\frac{1}{2}\mathbf{W}_{k,\cdot}\mathbf{Q}_{k}\mathbf{W}_{k,\cdot}^{T},
\end{align*}
and taking the conditional expectation with respect to $\boldsymbol{\beta}|\mathbf{Y},\boldsymbol{\theta}$
gives the result in Eq.~(\ref{eq:Hessian}).

\subsection{The anisotropic case}

We will in this subsection present the gradient and Hessian for the
anisotropic M$(2)$ model. We only consider $h_{0,k,x}=\log\left(h_{k,x}\right)$,
as the results for $h_{k,y}$ are completely symmetric. The derivations
can be performed analogously to Subsection~\ref{subsec:Gradient}
and Subsection~\ref{subsec:Approximate-Hessian}. The gradient is
\begin{align*}
\begin{split}
\frac{\partial\log p\left(\mathbf{y}|\boldsymbol{\theta}\right)}{\partial h_{0,k,x}}&=\text{tr}\left(\mathbf{K}_{k}^{-1}\frac{\partial\mathbf{K}_{k}}{\partial h_{0,k,x}}\right)-\tau_{k}^{2}\text{tr}\left(\tilde{\mathbf{Q}}^{-1}\left(\mathbf{J}^{kk}\otimes\left(\mathbf{K}_{k}\frac{\partial\mathbf{K}_{k}}{\partial h_{0,k,x}}\right)\right)\right)+\\ & \,\,\,\,\,\, -\tau_{k}^{2}\mathbf{M}_{k,\cdot}\mathbf{K}_{k}\frac{\partial\mathbf{K}_{k}}{\partial h_{0,k,x}}\mathbf{M}_{k,\cdot}^{T},
\end{split}
\end{align*}
where $\frac{\partial\mathbf{K}_{k}}{\partial h_{0,k,x}}=\exp\left(h_{0,k,x}\right)\mathbf{G}_{x}-\exp\left(-h_{0,k,x}\right)\exp\left(-h_{0,k,y}\right)\mathbf{G}_{z}$.
The approximate Hessian is
\begin{align*}
\begin{split}
E_{\boldsymbol{\beta}|\mathbf{Y},\boldsymbol{\theta}}\left[\frac{\partial^{2}\log p\left(\mathbf{y},\boldsymbol{\beta}|\boldsymbol{\theta}\right)}{\partial h_{0,k,x}^{2}}\right] & =-\text{tr}\left(\mathbf{K}_{k}^{-1}\frac{\partial\mathbf{K}_{k}}{\partial h_{0,k,x}}\mathbf{K}_{k}^{-1}\frac{\partial\mathbf{K}_{k}}{\partial h_{0,k,x}}\right)+\text{tr}\left(\mathbf{K}_{k}^{-1}\frac{\partial^{2}\mathbf{K}_{k}}{\partial h_{0,k,x}^{2}}\right)+\\
 & \,\,\,\,\,\,-\tau_{k}^{2}\left[\mathbf{M}_{k,\cdot}\mathbf{H}_{k,x}\mathbf{M}_{k,\cdot}^{T}+\text{tr}\left(\tilde{\mathbf{Q}}^{-1}\left(\mathbf{J}^{kk}\otimes\mathbf{H}_{k,x}\right)\right)\right],\nonumber 
\end{split}
\end{align*}
where $\frac{\partial^{2}\mathbf{K}_{k}}{\partial h_{0,k,x}^{2}}=\exp\left(h_{0,k,x}\right)\mathbf{G}_{x}+\exp\left(-h_{0,k,x}\right)\exp\left(-h_{0,k,y}\right)\mathbf{G}_{z}$
and \newline $\mathbf{H}_{k,x}=\frac{\partial\mathbf{K}_{k}}{\partial h_{0,k,x}}\frac{\partial\mathbf{K}_{k}}{\partial h_{0,k,x}}+\mathbf{K}_{k}\frac{\partial^{2}\mathbf{K}_{k}}{\partial h_{0,k,x}^{2}}$.
The first trace of this expression will be Hutchinson approximated
using $\text{tr}\left(\mathbf{K}_{k}^{-1}\frac{\partial\mathbf{K}_{k}}{\partial h_{0,k,x}}\mathbf{K}_{k}^{-1}\frac{\partial\mathbf{K}_{k}}{\partial h_{0,k,x}}\right)\approx\frac{1}{N_{s}}\sum_{j=1}^{N_{s}}\mathbf{v}_{j}^{T}\mathbf{K}_{k}^{-1}\frac{\partial\mathbf{K}_{k}}{\partial h_{0,k,x}}\mathbf{K}_{k}^{-1}\frac{\partial\mathbf{K}_{k}}{\partial h_{0,k,x}}\mathbf{v}_{j}$,
which requires solving two equation systems for every term and a somewhat
higher computational cost.

\subsection{Spatial hyperparameter priors\label{subsec:Spatial-parameter-priors}}

This section covers the priors of the spatial hyperparameters and
their derivatives and second derivatives. We drop the sub-indexing
with respect to $k$ throughout this section as the parameter priors
are mutually independent.

\subsubsection*{Priors for $\tau^{2}$ and $\kappa^{2}$ for M$(2)$}

For the spatial Mat\'{e}rn prior with $\alpha=2$ the joint PC log prior
for $\tau^{2}$ and $\kappa$ is
\begin{equation}
\log p\left(\tau^{2},\kappa\right)=-\frac{3}{2}\log\tau^{2}+\left(\frac{d}{2}-1-\nu\right)\log\kappa-\lambda_{1}\kappa^{d/2}-\lambda_{3}\kappa^{-\nu}\left(\tau^{2}\right)^{-1/2}+\text{const}.\label{eq:PCprior}
\end{equation}
The PC prior controls the spatial range $\rho$ and marginal variance
$\sigma^{2}$ of the spatial field, which are a bijective transform
of $\tau^{2}$ and $\kappa$, through a priori probabilities $P\left(\rho<\rho_{0}\right)=\xi_{1}$
and $P\left(\sigma^{2}>\sigma_{0}^{2}\right)=\xi_{2}$. This generates
the constants $\lambda_{1}=-\log\left(\xi_{1}\right)\left(\rho_{0}/\sqrt{8\nu}\right)^{d/2}$
and $\lambda_{3}=-\frac{\log\left(\xi_{2}\right)}{\sigma_{0}}\sqrt{\frac{\Gamma\left(\nu\right)}{\Gamma\left(\nu+d/2\right)\left(4\pi\right)^{d/2}}}$
in Eq.~(\ref{eq:PCprior}). As default we will use the values $\xi_{1}=\xi_{2}=0.05$,
$\rho_{0}=2$ voxels and $\sigma_{0}$ corresponding to $2\%$ of
the global mean signal. It is straightforward to obtain the derivatives
\begin{align}
\frac{\partial\log p\left(\tau^{2},\kappa\right)}{\partial\left(\tau^{2}\right)} & =-\frac{3}{2\tau^{2}}+\frac{\lambda_{3}\kappa^{-\nu}}{2}\left(\tau^{2}\right)^{-3/2}\label{eq:PCprior1derivative}\\
\frac{\partial\log p\left(\tau^{2},\kappa\right)}{\partial\left(\kappa^{2}\right)} & =\frac{1}{2\kappa}\left(\frac{d/2-1-\nu}{\kappa}-\frac{\lambda_{1}d}{2}\kappa^{d/2-1}+\lambda_{3}\nu\kappa^{-\nu-1}\left(\tau^{2}\right)^{-1/2}\right).\nonumber 
\end{align}
Changing parameterization to $\tau_{0}$ and $\kappa_{0}$ as before
and taking the second derivative with respect to these gives
\begin{align}
\frac{\partial^{2}\log p\left(\tau_{0},\kappa_{0}\right)}{\partial\tau_{0}^{2}}= & -\frac{\lambda_{3}}{4}\exp\left(-\frac{\nu\kappa_{0}}{2}\right)\exp\left(-\frac{\tau_{0}}{2}\right)\label{eq:PCprior2derivative}\\
\frac{\partial^{2}\log p\left(\tau_{0},\kappa_{0}\right)}{\partial\kappa_{0}^{2}}= & \frac{1}{4}\left[-\left(-d/2-1-\nu\right)\exp\left(-\frac{\kappa_{0}}{2}\right)-\left(\frac{d}{2}-1\right)\frac{\lambda_{1}d}{2}\exp\left(\frac{\kappa_{0}}{2}\left(\frac{d}{2}-1\right)\right)\right.\nonumber \\
 & \,\,\,\,+\left.\left(\nu-1\right)\lambda_{3}\nu\exp\left(\frac{\kappa_{0}}{2}\left(\nu-1\right)\right)\exp\left(-\frac{\tau_{0}}{2}\right)\right].\nonumber 
\end{align}

\subsubsection*{Priors for $\tau^{2}$ and $\kappa^{2}$ for M$\left(1\right)$}

In this situation, we use independent log-normal priors for $\tau^{2}$
and $\kappa^{2}$, that is $\tau_{0}\sim\mathcal{N}\left(\mu_{\tau_{0}},\sigma_{\tau_{0}}^{2}\right)$
and $\kappa_{0}\sim\mathcal{N}\left(\mu_{\kappa_{0}},\sigma_{\kappa_{0}}^{2}\right)$.
For $\tau_{0}$ we have the derivatives
\begin{align}
\frac{\partial\log p\left(\tau^{2},\kappa\right)}{\partial\tau_{0}} & =-\frac{\tau_{0}-\mu_{\tau_{0}}}{\sigma_{\tau_{0}}^{2}},\,\,\,\,\,\,\,\frac{\partial^{2}\log p\left(\tau^{2},\kappa\right)}{\partial\tau_{0}^{2}}=-\frac{1}{\sigma_{\tau_{0}}^{2}},\label{eq:logNPriorM1}
\end{align}
and correspondingly for $\kappa_{0}$. Per default we use $\mu_{\tau_{0}}=\log0.01$,
$\mu_{\kappa_{0}}=\log0.1$, $\sigma_{\tau_{0}}=4$, $\sigma_{\kappa_{0}}=1$.

\subsubsection*{Priors for $h_{x}$ and $h_{y}$ for the anisotropic prior}

We use a log-normal prior for $h_{x}$ and $h_{y}$.
For $h_{0,x}=\log\left(h_{x}\right)$ we have the derivatives
\begin{align}
\frac{\partial\log p\left(h_{0,x},h_{0,y}\right)}{\partial h_{0,x}} & =-\frac{2}{3\sigma_{h}^{2}}\left(h_{0,x}-\frac{h_{0,y}}{2}\right),\,\,\,\,\,\,\,\frac{\partial^{2}\log p\left(h_{0,x},h_{0,y}\right)}{\partial h_{0,x}^{2}}=-\frac{2}{3\sigma_{h}^{2}},\label{eq:logNPriorAniso}
\end{align}
and correspondingly for $h_{0,y}$.

\subsubsection*{Priors for $\tau^{2}$ for ICAR$\left(1\right)$ and ICAR$\left(2\right)$}

We use the PC prior for $\tau^{2}$ for Gaussian random effects from
\citet{Simpson2017}
\begin{equation}
p\left(\tau^{2}\right)=\frac{\lambda_{2}}{2}\left(\tau^{2}\right)^{-3/2}\exp\left(-\lambda_{2}\left(\tau^{2}\right)^{-1/2}\right),\,\,\,\,\,\,\tau^{2}>0.\label{eq:PCpriortau2}
\end{equation}
By specifying $\sigma_{0}^{2}$ and $\xi_{2}$ so that $P\left(\sigma_{i|-i}^{2}>\sigma_{0}^{2}\right)=\xi_{2}$,
we get $\lambda_{2}=-\log\left(\xi_{2}\right)/\left(\sigma_{0}\sqrt{6}\right)$
for ICAR$\left(1\right)$ and $\lambda_{2}=-\log\left(\xi_{2}\right)/\left(\sigma_{0}\sqrt{42}\right)$
for ICAR$(2)$. Derivatives are obtained as
\begin{align}
\frac{\partial\log p\left(\tau^{2}\right)}{\partial\left(\tau^{2}\right)} & =-\frac{3}{2\tau^{2}}+\frac{\lambda_{2}}{2}\left(\tau^{2}\right)^{-3/2}\label{eq:PCpriorderivativetau2}\\
\frac{\partial^{2}\log p\left(\tau_{0}\right)}{\partial\tau_{0}^{2}}= & -\frac{\lambda_{2}}{4}\exp\left(-\frac{\tau_{0}}{2}\right).\nonumber 
\end{align}
When comparing to results from our older paper \citep{Sid??n2017},
we use the same gamma prior as used there for ICAR$\left(1\right)$,
$\tau^{2}\sim\Gamma\left(q_{1},q_{2}\right)$, which has derivatives
\begin{align*}
\frac{\partial\log p\left(\tau^{2}\right)}{\partial\left(\tau^{2}\right)} & =q_{2}-1-\frac{\tau^{2}}{q_{1}}\\
\frac{\partial^{2}\log p\left(\tau_{0}\right)}{\partial\tau_{0}^{2}}= & -\frac{1}{q_{1}}\exp\left(\tau_{0}\right).
\end{align*}
We use the default values $q_{1}=10$ and $q_{2}=0.1$.

\section{Cross-Validation\label{sec:Cross-Validation}}

\subsection{Cross-validation over left out voxels}

To reduce the impact of the noise model, we drop the entire BOLD time
series for $d\%$ of the voxels to obtain the cross-validation voxel
set $D$, and predict the time series $\mathbf{Y}_{\cdot,D}$ given
the other data $\mathbf{Y}_{\cdot,D^{c}}$. To separate the activation
and nuisance regressors, we rewrite the model as
\begin{equation}
\underset{T\times N}{\mathbf{Y}}=\underset{T\times K}{\mathbf{X}}\underset{K\times N}{\mathbf{W}}+\mathbf{E}=\left[\begin{array}{cc}
\dot{\mathbf{X}} & \ddot{\mathbf{X}}\end{array}\right]\left[\begin{array}{c}
\dot{\mathbf{W}}\\
\ddot{\mathbf{W}}
\end{array}\right]+\mathbf{E}=\dot{\mathbf{X}}\dot{\mathbf{W}}+\ddot{\mathbf{X}}\ddot{\mathbf{W}}+\mathbf{E},\label{eq:ModelDivided}
\end{equation}
with $\dot{\mathbf{W}}$ corresponding to the $K_{act}\times N$ activity
regression coefficients, which have spatial priors, and $\ddot{\mathbf{W}}$
corresponding to the nuisance regression coefficients. We define the
cross-validation error time series as
\begin{align}
\mathbf{E}_{\cdot,D}^{CV} & =\mathbf{R}_{\cdot,D}-\ddot{\mathbf{X}}E\left(\ddot{\mathbf{W}}_{\cdot,D}|\mathbf{R}_{\cdot,D},\boldsymbol{\theta}\right),\label{eq:CVerror}\\
\mathbf{R}_{\cdot,D} & =\mathbf{Y}_{\cdot,D}-\dot{\mathbf{X}}E\left(\dot{\mathbf{W}}_{\cdot,D}|\mathbf{Y}_{\cdot,-D},\boldsymbol{\theta}\right),\nonumber 
\end{align}
which is computed in two steps. First the out-of-sample residuals $\mathbf{R}_{\cdot,D}$
of the spatial part of the model are computed and then a new model
$\mathbf{R}_{\cdot,D}=\ddot{\mathbf{X}}\ddot{\mathbf{W}}_{\cdot,D}+\mathbf{E}_{\cdot,D}$
is fitted for each voxel independently, using the original values
for the parameters $\boldsymbol{\theta}$, before the error time series
$\mathbf{E}_{\cdot,D}^{CV}$ can be computed. The reason for this
seemingly complicated procedure is to reduce the impact of the nuisance
regressors on the evaluation of the spatial model. We compute the
in-sample errors as $\mathbf{E}^{IS}=\mathbf{Y}-\mathbf{X}E\left(\mathbf{W}|\mathbf{Y},\boldsymbol{\theta}\right)$.
We compute the MAE and RMSE for voxel set $D$ as
\begin{equation}
\text{MAE}=\frac{1}{T\left|D\right|}\sum_{t=1}^{T}\sum_{n\in D}\left|\mathbf{E}_{t,n}\right|,\,\,\,\,\,\text{RMSE}=\sqrt{\frac{1}{T\left|D\right|}\sum_{t=1}^{T}\sum_{n\in D}\mathbf{E}_{t,n}^{2}}.\label{eq:MAERMSE}
\end{equation}
The spatial posterior predictions for the dropped voxels $E\left(\dot{\mathbf{W}}_{\cdot,D}|\mathbf{Y}_{\cdot,-D},\boldsymbol{\theta}\right)$
are computed using Eq.~(\ref{eq:postbeta_ar}) after replacing $\tilde{\mathbf{Q}}_{n}$
and $\tilde{\mathbf{q}}_{n}$ with $\mathbf{0}$ for all $n\in D$.

For the proper scoring rules CRPS, IGN and INT, we also need to compute
the predictive standard deviation in each voxel. This is done in a
way that neglects the uncertainty in the intercept and head motion
regressors. For an unseen datapoint $\tilde{\mathbf{Y}}_{t,n}$, the
law of total variance gives
\begin{align}
\text{Var}\left(\tilde{\mathbf{Y}}_{t,n}|\mathbf{Y}_{\cdot,-D},\boldsymbol{\theta},\ddot{\mathbf{W}}\right) & =\text{\ensuremath{\text{E}_{\mathbf{\dot{\mathbf{W}}}|\mathbf{Y}_{\cdot,-D}}\left[\text{Var}\left(\tilde{\mathbf{Y}}_{t,n}|\boldsymbol{\theta},\mathbf{W}\right)\right]}}+\text{Var}_{\mathbf{\dot{\mathbf{W}}}|\mathbf{Y}_{\cdot,-D}}\left[\text{E\ensuremath{\left(\tilde{\mathbf{Y}}_{t,n}|\boldsymbol{\theta},\mathbf{W}\right)}}\right]\label{eq:predictiveVariance}\\
 & =\text{Var\ensuremath{\left(\tilde{\mathbf{Y}}_{t,n}|\boldsymbol{\theta},\mathbf{W}\right)}}+\mathbf{X}_{t,(1:K_{act})}\text{Var}\left(\dot{\mathbf{W}}|\mathbf{Y}_{\cdot,-D}\right)\mathbf{X}_{t,(1:K_{act})}^{T},\nonumber 
\end{align}
where the first term is simply the variance of the AR noise process
in voxel $n$ that does not depend on $\mathbf{W}$ and which can
be obtained given the AR parameters $\mathbf{A}_{\cdot,n}$ through
the Yule-Walker equations \citep[see for example][]{Cryer2008}. The
second term can be computed using the simple RBMC estimator as in Eq.~(\ref{eq:simpleRBMC}) in the main article after replacing $\tilde{\mathbf{Q}}_{n}$
with $\mathbf{0}$ for all $n\in D$ as was done for the mean. The
predictive distribution for $x=\mathbf{E}_{t,n}$ is Gaussian with
mean $\mu=0$ and variance $\sigma^{2}$ as in Eq.~(\ref{eq:predictiveVariance})
which gives simple expressions for the scores as
\begin{align*}
\text{CRPS}_{t,n} & =\sigma\left[\frac{1}{\sqrt{\pi}}-2\varphi\left(\frac{x-\mu}{\sigma}\right)-\frac{x-\mu}{\sigma}\left(2\Phi\left(\frac{x-\mu}{\sigma}\right)-1\right)\right],\\
\text{IGN}_{t,n} & =\log\left(\frac{1}{\sigma}\varphi\left(\frac{x-\mu}{\sigma}\right)\right),\nonumber \\
\text{INT}_{t,n} & =2A\sigma+\frac{2}{u}\left[\left(\mu-A\sigma-x\right)\mathbf{1}\left(x<\mu-A\sigma\right)+\left(x-\left(\mu+A\sigma\right)\right)\mathbf{1}\left(x>\mu+A\sigma\right)\right],\nonumber 
\end{align*}
where $\varphi$ and $\Phi$ denotes the standard normal PDF and CDF
and $A=\Phi^{-1}\left(1-u/2\right)\approx1.96$ for $u=0.05$,
which is used by default. The presented values for the scores are
averages across all time points and left out voxels.

\section{Additional Results}
\begin{figure}[H]
\includegraphics[width=1\linewidth,trim={1mm 8mm 2mm 4mm},clip]{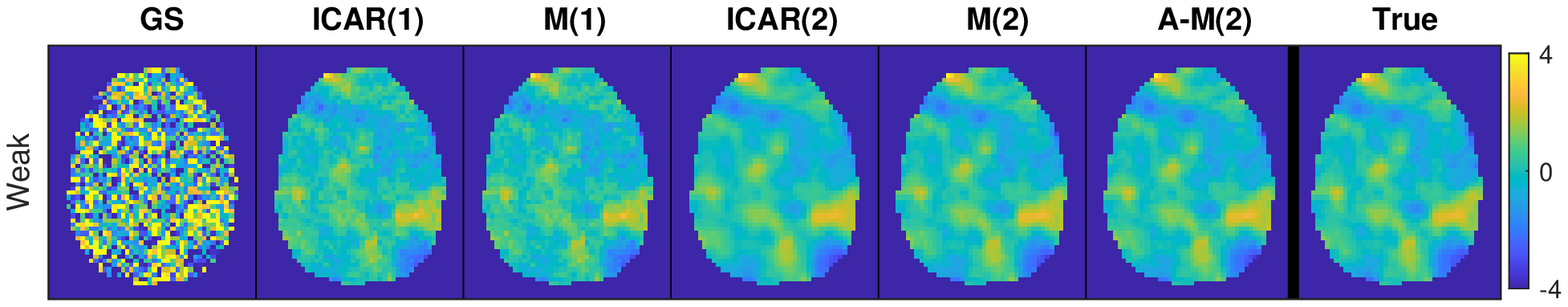}
\includegraphics[width=1\linewidth,trim={1mm 8mm 2mm 9mm},clip]{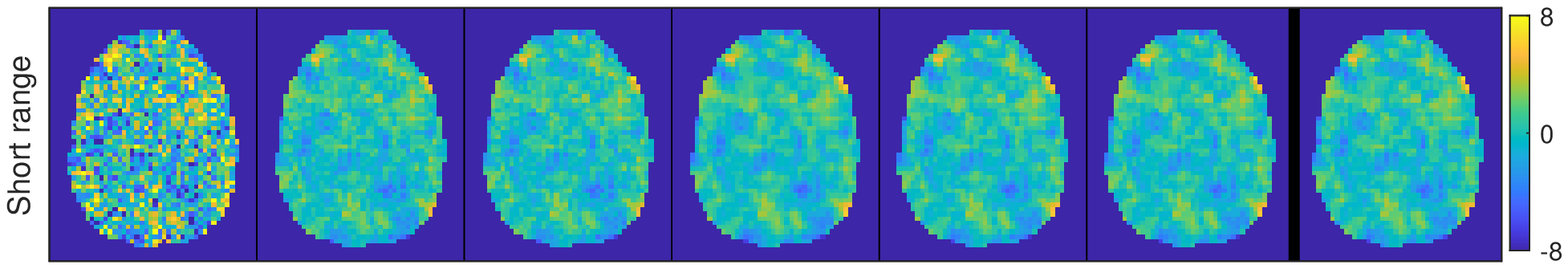}
\includegraphics[width=1\linewidth,trim={1mm 8mm 2mm 9mm},clip]{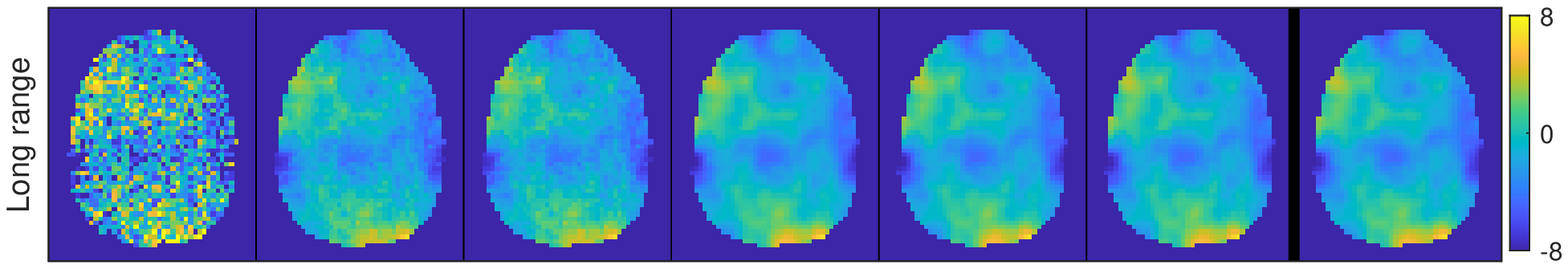}
\includegraphics[width=1\linewidth,trim={1mm 8mm 2mm 9mm},clip]{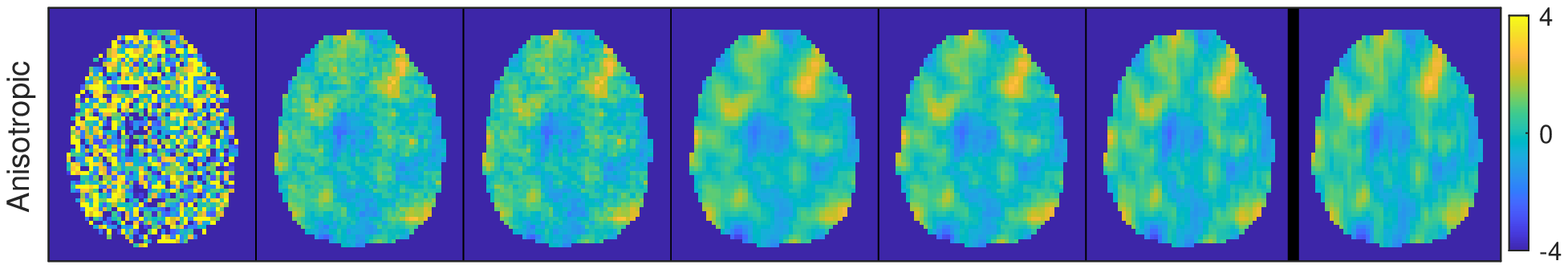}

\caption{Posterior means of activity coefficients for the four conditions of the simulated dataset, estimated with different spatial priors, which are used in computation of the PPMs in Fig.~\ref{fig:simulatedMaps}.}
\label{fig:simulatedPostMeanMaps}
\end{figure}

\begin{figure}[H]
\includegraphics[width=1\linewidth,trim={1mm 6mm 2mm 2mm},clip]{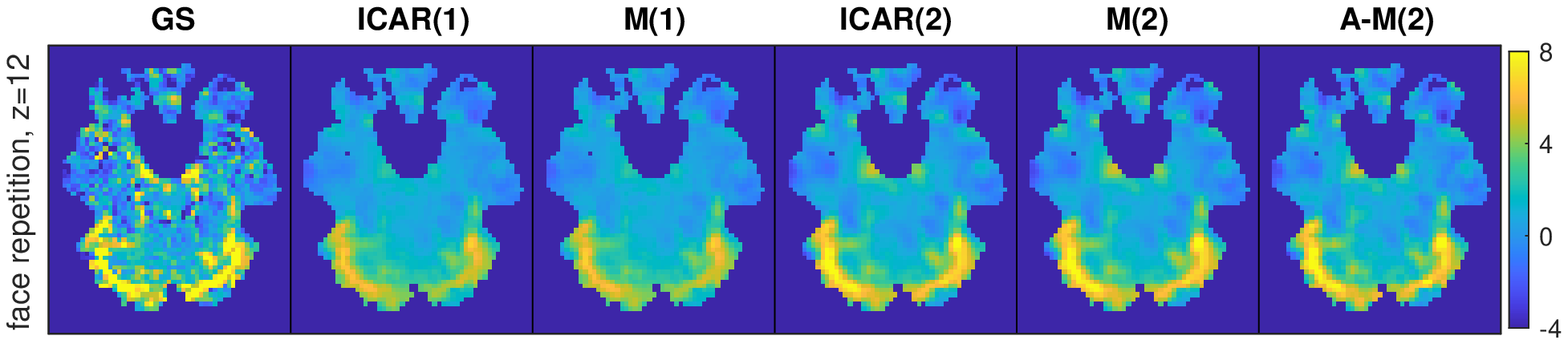}
\includegraphics[width=1\linewidth,trim={1mm 6mm 2mm 7mm},clip]{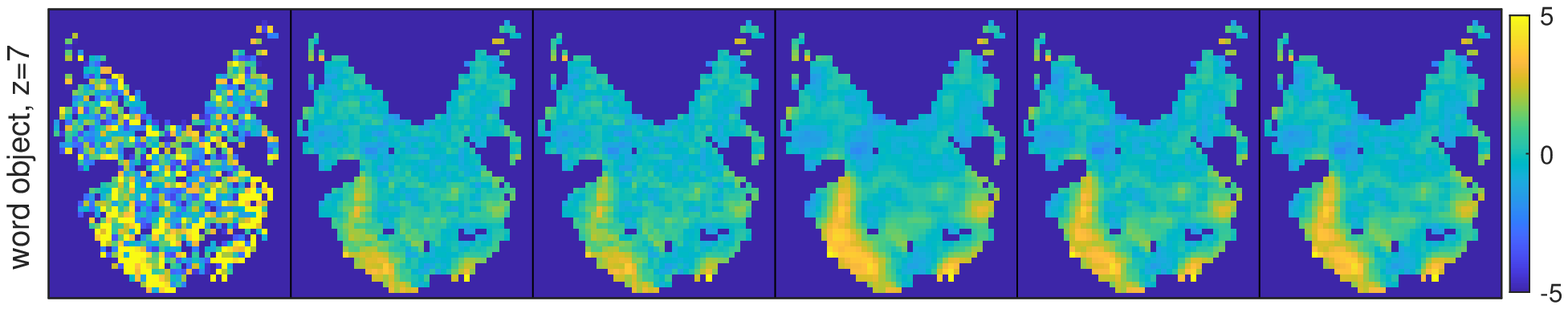}
\includegraphics[width=1\linewidth,trim={1mm 6mm 2mm 10mm},clip]{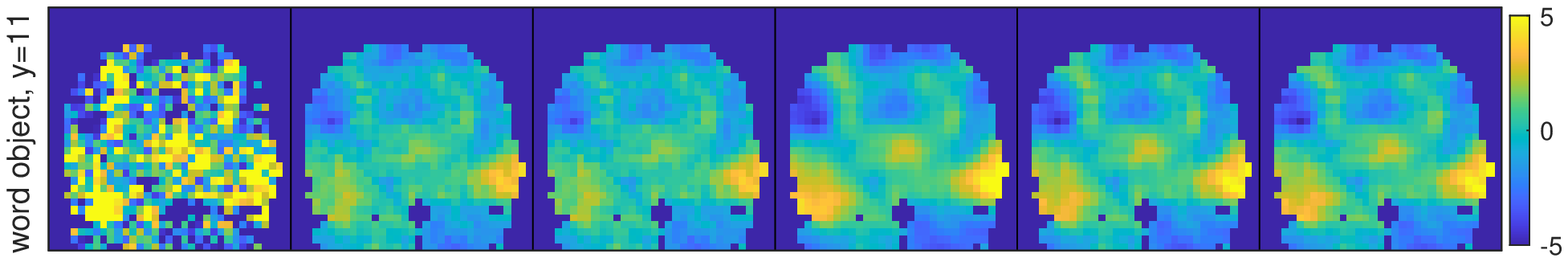}

\caption{Posterior means for the two real datasets, when using different spatial priors, which are used in computation of the PPMs in Fig.~\ref{fig:realPPMs}. The top row shows axial slice 12 of the face repetition dataset, and the middle and bottom rows show axial slice 7 and coronal slice 11 of the word object dataset. The spatial priors are summarised in Table~\ref{tab:Precision-matrices}.}
\label{fig:realPostMeans}
\end{figure}

\begin{table}[H]
\caption{Cross-validation scores for the two datasets, comparing the different
spatial priors. The scores are computed as means across voxels, and presented in negatively oriented forms, so that smaller values are always better. In-sample refers
to the average across all voxels, while the other columns shows means
and standard errors across 50 random sets of left out voxels. \label{tab:CVTable}}
\includegraphics[trim=1mm 0mm 0mm 1mm,clip,scale=0.8]{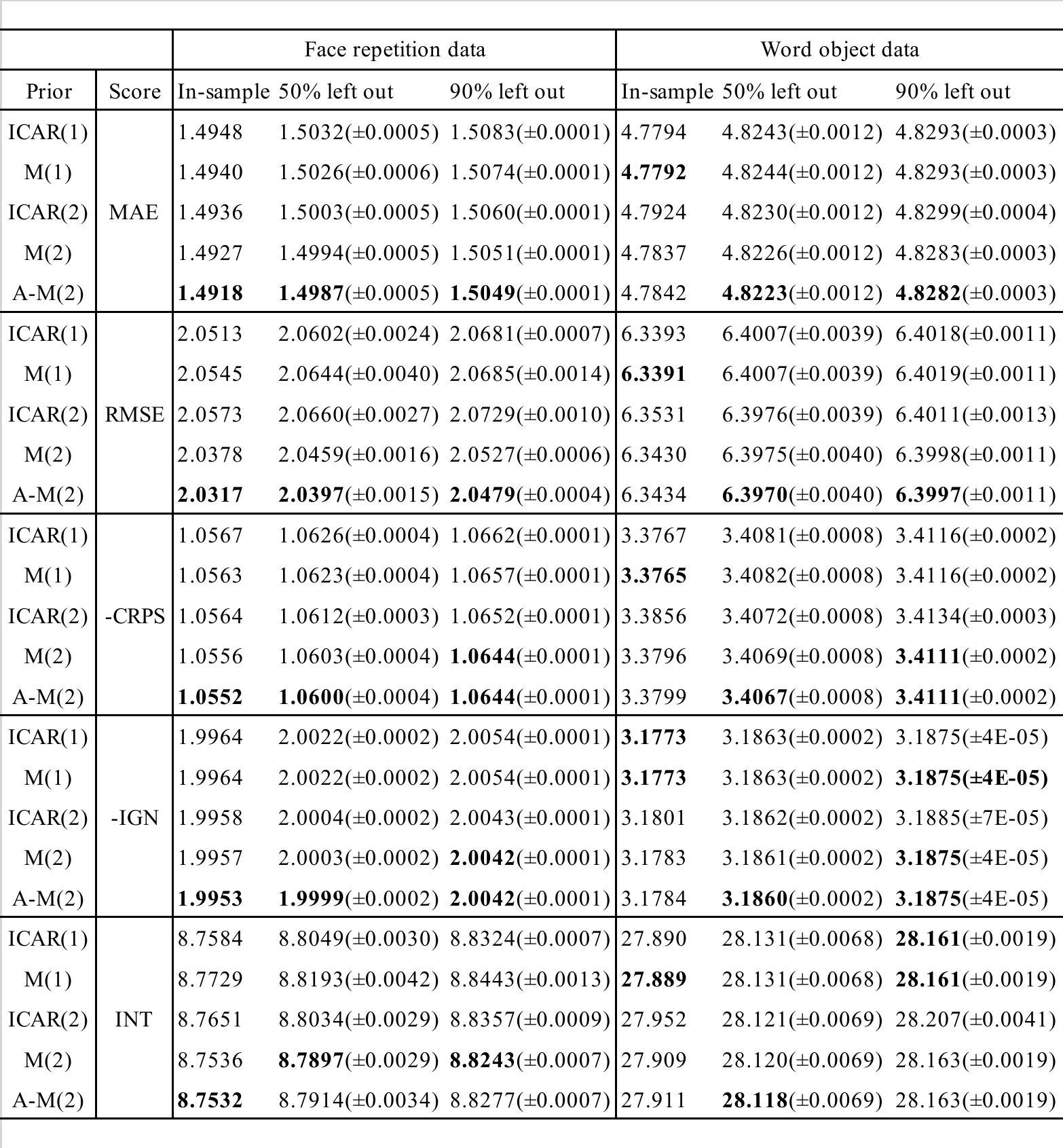}
\end{table}
\end{supplement}

\end{document}